\documentclass[journal]{IEEEtran}
\usepackage[dvipsnames]{xcolor}
\usepackage[cmex10]{amsmath}
\usepackage{graphicx,times,color,amssymb,cite}
\usepackage{epstopdf}
\usepackage[disable]{todonotes} 
\usepackage{amsthm}
\usepackage{fancyvrb}
\newtheorem{theorem}{Theorem}
\newtheorem{assumption}{Assumption}
\newtheorem{lemma}{Lemma}

\newtheorem{definition}{Definition}

\def\dh{\hat{d}}

\def\ee{{\epsilon}}

\def\be{\begin{equation}}
\def\bea{\begin{eqnarray}}
\def\beas{\begin{eqnarray*}}
\def\eea{\end{eqnarray}}
\def\eeas{\end{eqnarray*}}

\def\bi{\begin{itemize}}
\def\ba{\begin{assumption}}
\def\ea{\end{assumption}}
\def\ee{\end{equation}}
\def\ei{\end{itemize}}

\def\xh{\hat{x}}

\def\xt{\tilde{x}}

\def\cA{{\cal A}}
\def\cG{{\cal G}}
\def\cF{{\cal F}}

\def\xb{\bar{x}}
\def\cR{{\cal R}}
\def\cC{{\cal C}}
\def\cS{{\cal S}}
\def\cD{{\cal D}}
\def\cE{{\cal E}}

\def\cU{{\cal U}}

\def\cP{{\cal P}}

\def\cU{{\cal U}}

\def\cN{{\cal N}}

\def\cN{{\cal N}}

\def\cC{{\cal C}}
\def\cS{{\cal S}}
\def\cD{{\cal D}}

\def\cA{{\cal A}}
\def\cU{{\cal U}}

\def\bmat{\begin{matrix}}

\def\emat{\end{matrix}}

\ifCLASSINFOpdf
\else
\fi
\usepackage[tight,footnotesize]{subfigure}

\usepackage{url}

\begin{document}
\bstctlcite{IEEEexample:BSTcontrol}
	
%
\title{Stability and Resilience of Distributed Information Spreading  in Aggregate Computing}
%
%
%

\author{Yuanqiu Mo, \IEEEmembership{Member IEEE}, Soura Dasgupta,  \IEEEmembership{Fellow IEEE} and Jacob Beal, \IEEEmembership{Senior Member IEEE}
\thanks{
%
This work has been supported by the Defense Advanced Research Projects Agency (DARPA) under Contract No. HR001117C0049.
The views, opinions, and/or findings expressed are those of the author(s) and should not be interpreted as representing the official views or policies of the Department of Defense or the U.S. Government.
This document does not contain technology or technical data controlled under either U.S. International Traffic in Arms Regulation or U.S. Export Administration Regulations.
Approved for public release, distribution unlimited (DARPA DISTAR case 33519, 10/14/20).}
\thanks{Y. Mo is with the Institute of Advanced Technology, Westlake Institute for Advanced Study, Westlake University, Hangzhou 310024, China (email: moyuanqiu@@westlake.edu.cn). 
	
S. Dasgupta is with the University of Iowa, Iowa City, Iowa  52242 USA (e-mail: soura-dasgupta@uiowa.edu).
S. Dasgupta is also a Visiting Professor at Shandong Computer Science Center, Shandong Provincial Key Laboratory of Computer Networks, China.}
\thanks{J. Beal is with Raytheon BBN Technologies,
Cambridge, MA, USA 02138 USA (e-mail: jakebeal@ieee.org)}
}

\maketitle

\begin{abstract}
Spreading information through a network of devices is a core activity for most distributed systems.
As such, self-stabilizing algorithms implementing information spreading are one of the key building blocks enabling aggregate computing to provide resilient coordination in open complex distributed systems.
This paper improves a  general spreading block in the aggregate computing literature by making it resilient to network perturbations, establishes its global uniform asymptotic stability and proves that it is ultimately bounded under persistent   disturbances.  
The ultimate bounds depend only on the magnitude of the largest perturbation and the network diameter, 
and three  design parameters  trade off  competing aspects of performance.
For example, as in many dynamical systems, values leading to greater resilience to network perturbations slow convergence and vice versa.
\end{abstract}


%
\IEEEpeerreviewmaketitle

\section{Introduction}

Complex networked distributed systems  are rapidly becoming a feature of many engineering systems. 
Their stability, dynamics and reliability  has acquired paramount importance. 
Control theorists have embraced this challenge through extensive research on the stability of  networked control systems, most typically ``closed'' systems in which a good model of the system is available at design time e.g.~\cite{large,suri,summers2009formation,summers2011control,arcak2007passivity,hatanaka2015passivity,guler2016adaptive,fidan2013adaptive,formation}.

An emerging set of  alternative challenges is  posed by the analysis and design of  complex open systems like tactical information sharing, smart cities,  edge computing, personal and home area networks, and the Internet of Things (IoT)~\cite{BPV-Computer15}.
These systems also disperse  services to local devices, require devices to interact \emph{safely and seamlessly} with nearby brethren  through  peer to peer information exchange, and to share tasks. 
As they are \emph{open}, however, they must support frequent  non-centralized changes in the applications and services that they host.
Current modes of device interactions restrict their  potential  by being  typically either highly constrained and inflexible (e.g. single-purpose devices) or  by relying on remote infrastructure like cloud services. 
The former impairs reusability  and prevents  devices from contributing to multiple overlapping applications. The latter is  centralized with high latency and lacks the agility to exploit local communication, services and devices.

\emph{Aggregate computing}, on the other hand, addresses device coordination in open systems with a layered  approach~\cite{BPV-Computer15}, separating systems into abstraction layers that decomposing systems engineering into separable tasks, much like the OSI model does for communication networks, \cite{OSI}.
The layers span from applications to a {\em field calculus} (providing distributed scoping of shared information), and an abstract device model (for services such as neighborhood discovery). 
Between these, a middle layer facilitates resilient device interactions and comprises three classes of basis set modules that are themselves distributed graph algorithms: 
(i) $G$-blocks that spread information through a network of devices, 
(ii) $C$-blocks that summarize salient information about the network to be used by interacting units,
and (iii) $T$-blocks that maintain temporary state.   
Prior work \cite{BV-FOCAS14,VD-COORD2014-LNCS2014,VABDP-ACM},   has  shown \emph{that a broad class of device interactions can be realized by  interconnections  of these three blocks}. Our long term research goal is analyze the stability of compositions of these blocks including in feedback. 
This paper is concerned with a rigorous analysis of the most general $ G $-block of \cite{VABDP-ACM}, \emph{after making it resilient to network perturbations}.

While the empirical assessment of compositions of the $ G $-block with other blocks is encouraging\cite{KumarBeal15,VBDP-SASO15,CRF}, the formal analysis of its most general case has been confined to  self-stablization \cite{VABDP-ACM}, and that under the assumption that all states lie in Noetherian rings and are thus \emph{a priori bounded}. 
Unlike global uniform asymptotic stability (GUAS) \cite{khalil,Hahn}, self-stabilization has no notion of robustness to  perturbations, while perturbations are to be expected in any feedback composition. 
Thus we improve the generalized $ G $-block to allow removal of the Noetherian assumption, proof of GUAS,
and (under an additional Lipschitz condition) ultimate boundedness in face of persistent  perturbations.
Finding ultimate bounds further anticipates the  development of  sophisticated variants  of the small gain  theorem~\cite{smallgain}, \cite{khalil} for  closed loop analysis.

Previously, in \cite{cdc16} and \cite{TAC}, we have performed such an analysis of the simplest $G$-block, the Adaptive Bellman-Ford (ABF) algorithm, which estimates distances of nodes from a source set in a distributed manner and (unlike the classical Bellman-Ford algorithm~\cite{bellman1958routing}) accommodates underestimates.
In \cite{mo2018robust} we have analyzed \emph{without proof} another special case, which 
  generalizes ABF by allowing  non-Euclidean distance metrics, e.g., those that penalize
certain routes, and permits \emph{broadcast}, where each source broadcasts a  value like a diameter estimate it holds to nearest devices.  More features are given  Section~\ref{salg}.

A further problem that must be considered in this context is the \emph{rising value problem}.
All $ G $-block algorithms generate estimates $ \xh_i(t) $ that must converge to a value $ x_i $. 
The rising value problem is when underestimates ( i.e. $ \xh_i(t)<x_i $) may rise very slowly, and this problem affects the $ G $-blocks analyzed in \cite{TAC} and \cite{mo2018robust}.
The most general $ G $-block given in \cite{VABDP-ACM} removes this problem by treating the estimate generated by the algorithm in \cite{mo2018robust} as an auxiliary state   $ \xt_i(t) $. The actual  estimate $ \xh_i(t) $ is increased by a fixed amount of at least $ \delta>0 $ unless  $ \xt_i(t+1) $  equals either $ \xh_i(t) $ or  the maximal element in the Noetherian ring.  If either of these two conditions is violated then $ \xh_i(t+1)=\xt_i(t+1). $ The increase by $ \delta $ or more, removes the rising value problem. However the equality requirement of $ \xt_i(t+1)=\xh_i(t) $ introduces fragility to  disturbances, since $ \xt(t+1)=\xh_i(t) $ rarely occurs under perturbations and thus $ \xh_i(t) $ must persistently rise to the maximal element.

We deal with real non-negative numbers rather than Noetherian rings and \emph{do not assume prior estimate bounds}.  Instead we modify the algorithm in \cite{VABDP-ACM} by introducing two additional parameters, a {\em modulation threshold} $ M $ and  a {\em dead zone} value $ D $, that together define regions for improved perturbation tolerance versus regions for fast convergence with $\delta$, and that reduce to the algorithm in \cite{mo2018robust} when $ M=0 $ and/or $ D=\infty. $
We show that the improved algorithm is GUAS for all non-negative $ M $ and $ D $ \emph{without the assumption} that $ M $ \emph{is a maximal element}. 
In the case of persistent bounded perturbations, we show that the estimates are ultimately bounded provided that the dead zone parameter $ D $ exceeds a value proportional to the disturbance bound. 
A larger $D$, however, is also less effective at mitigating the rising value problem, indicating a trade-off between speed and robustness that is common to most dynamical systems. 

In the remainder of the paper, Section \ref{salg} provides the algorithm,  assumptions and  motivating applications. Section \ref{sstat} characterizes stationary points, while Section \ref{guas} proves GUAS. 
Section \ref{robust} gives ultimate bounds \emph{which are determined only by the magnitude of the perturbations and the graph.} 
Section \ref{sdesign} discusses design choices, which are affirmed via simulation in Section \ref{ssim}, and Section \ref{sconc} concludes.

\section{Algorithm} \label{salg}
In this section, we present a general $ G $-block that spreads information through a network in a distributed fashion. 
Originally provided in \cite{viroli2018engineering} using the language of field calculus, we translate it here to one more appropriate of this journal. 
Section \ref{scdc}  describes a special case shown to be GUAS in \cite{mo2018robust}, with proofs omitted, plus examples and a shortcoming.
Section \ref{sgalg} then presents a more general algorithm that removes this deficiency, and Section \ref{sdef} provides assumptions and  definitions that will be used for proofs in subsequent sections.

\subsection{The Spreading block of\cite{mo2018robust}}\label{scdc}
Consider an undirected graph $\cG = (V,E)$ with nodes  in  $V=\{1,2,\cdots, N\}$  and edge set $E$. Nodes $i$ and $k$ are \emph{neighbors} if they share an edge.
The goal of the algorithm is to spread the state $ x_i $ to node $ i. $ Denote   $\cN(i)$ as the set of neighbors of $i$. With $\xh_i(t)$  an estimate of  $ x_i $, in the $ t $-th iteration,  the information spreading  in \cite{mo2018robust}, proceeds as:
\begin{equation}
\xh_i(t+1)=\min\left\lbrace \min_{k\in \cN(i)}\left\lbrace f\left (\xh_k(t), e_{ik} \right ) \right\rbrace, s_i  \right\rbrace, \forall t\geq t_0 .  \label{spreading}
\end{equation}

The $e_{ik}$ define the structural aspects of $\cG$; e.g. they may be the edge lengths between neighbors; $s_i\geq 0$, which may be either finite or infinite, is the \emph{maximum value} that $\xh_i(t)$ can acquire after the initial time.
Further, $ \xh_i(t_0)\geq 0, $ for all $ i\in V. $

Function $f(\cdot,\cdot)$ must be \emph{progressive}  i.e. for some $\sigma>0$,
\begin{equation}\label{sigma}
f(a,b)>a+\sigma
\end{equation}
and \emph{monotonic} in the first variable,  i.e.
\begin{equation}\label{monotonic}
f(a_1,b)\geq f(a_2,b), \mbox{ if } a_1\geq a_2.
\end{equation}
and is finite for  finite $a$ and $b$.  The initialization in  (\ref{spreading}) ensures that $ \xh_i(t)\geq 0 $, for all $ t\geq t_0. $
Define  $\cS^*$ as the set of  nodes with \emph{finite maximum values} $ s_i $:
\begin{equation}\label{source}
\cS^* = \{i \in V |s_i < \infty \}.
\end{equation}
We will assume that this set is nonempty. Evidently, the information $ x_i $ to be spread to node $ i $ must be  an element of the vector of \emph{stationary} values of (\ref{spreading}),  i.e. obeys
\begin{equation}\label{stationary}
x_i=\min\left\lbrace \min_{k\in \cN(i)}\left\lbrace f\left (x_k, e_{ik} \right ) \right\rbrace, s_i  \right\rbrace, \forall ~i\in V .
\end{equation}
We will prove the less than evident fact that this stationary point is unique,  finite, and that at  least one $ x_i=s_i$.

The simplest example,  whose Lyapunov analysis is  in \cite{cdc16}, \cite{TAC}, is ABF where $ f(a,b)=a+b $ and  $ s_i=0 $ or infinity. The set of $ i $ for which $ s_i=0 $ are called sources, $ e_{ik}>0 $ is the edge length between nodes $ i $ and $ k $,  $ x_i $ represents the  distance $ d_i $ from the set of sources, and with $ \xh_i=\dh_i $, (\ref{spreading}) becomes 
\begin{equation}\label{abf}
\dh_i(t+1)=\left \{ \begin{matrix}
0 & s_i=0\\
\min_{k\in \cN(i)}\left \{\dh_k(t)+e_{ik}\right \} &s_i\neq 0
\end{matrix}\right ..
\end{equation}

\begin{figure}[h]
	\centering
	\includegraphics[width = 0.7\columnwidth]{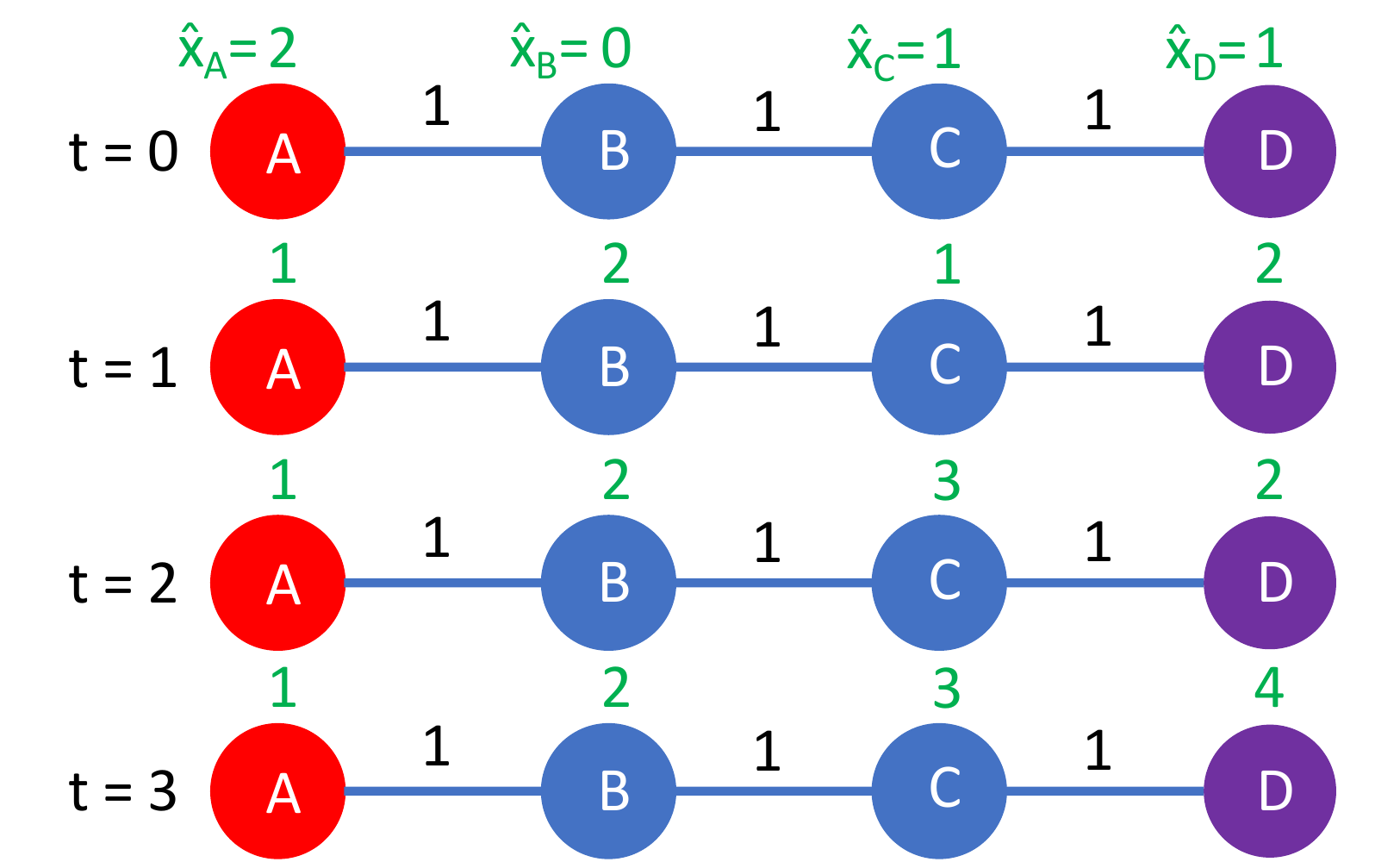}
	\caption{Non-zero $s_i$ representing external gateways of a tactical wireless network: red and purple nodes are high and low speed external links, respectively, while blue are nodes without external links. Black numbers represent edge weights $e_{ik}$, green numbers represent state estimates $\xh_i$. After 3 rounds, all nodes, including the low-speed link, have converged to route through the high-speed link.  }
	\label{fig:two_links}
\end{figure}

We may observe that an ABF stationary point must have all nodes with $ s_i=0 $ also having $ d_i=0. $
If we allow other values of $s_i$, then not all members of $ \cS^* $ necessarily have stationary states $ s_i $.
Figure \ref{fig:two_links} shows one such example, in which $s_i$ represents transport lag for external gateways in a tactical wireless network.
In this example, as with (\ref{abf}),  
\begin{equation}\label{sum}
f(\xh_k(t), e_{ik}) = \xh_k(t) + e_{ik}.
\end{equation} 
Here, node A (red) is a high-speed gateway with $s_A = 1$, node D (purple) is a low-speed gateway with $s_D = 5$, and the others are non-gateways with $s_i = \infty$.
Through (\ref{spreading}), all nodes try to route to  external networks through the shortest effective path. After 3 rounds, all nodes, including the low-speed link,  converge to route  through the high-speed link. In this case the stationary state of the low speed link is $ x_D=4 $ and does not equal $ s_i=5 $ even though $ s_i $ is finite. Should the high speed link represented by node A disappear,  then the state estimate of D does converge to its maximum value 5, while those of nodes B and C converge to 7 and 6, respectively, i.e. nodes reroute through the still available low speed link.
\begin{figure}[h]
	\centering
	\includegraphics[width = 1\columnwidth]{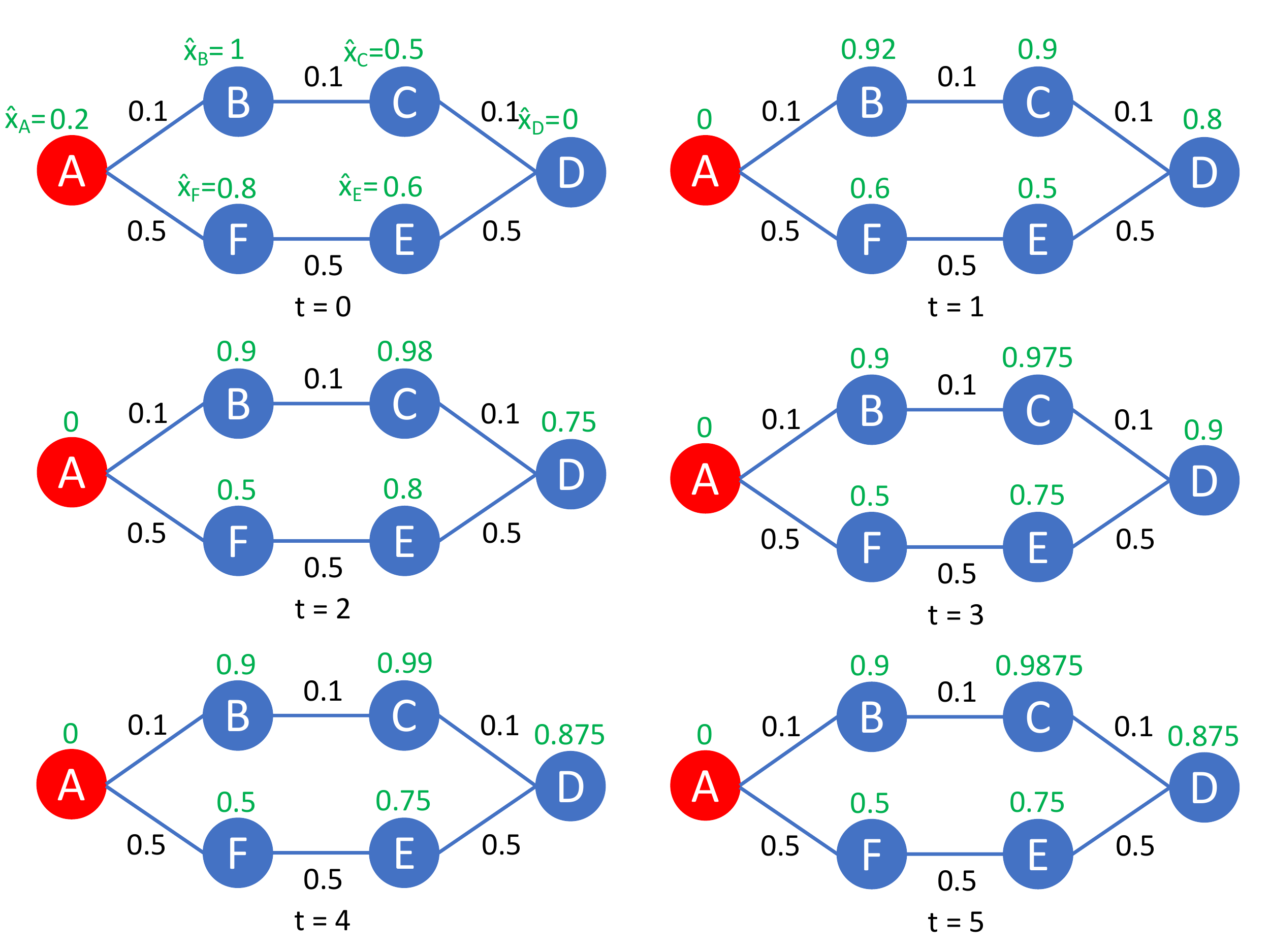}
	\caption{Illustration of the spreading block. In this example, node in red represents the source with a maximum value of 0, while nodes in blue have a maximum value of 1. The edge value $e_{ik}$ reflects the success rate of delivery. After 5 rounds, each node finds a path with the smallest failure rate of delivery to the source. }
	\label{fig:probability}
\end{figure}

While in the prior examples $ f(a,b) $ is linear and increasing in   $ b $, this need not be the case. 
A specialization of (\ref{spreading}) violating both these properties  finds the   \emph{most probable path}  (MPP) from each node in a network to a source. In this case $ e_{ik} $ represents the probability of successful traversal or delivery between neighbors $ i $ and $ j $. The stationary   value $ x_i $ is the \emph{smallest} probability of failure in movement from  node $ i $ to the source. 
In this case $ x_i=0 $ for sources. For all other nodes  
\[ x_i=\min_{k\in \cN(i)}\{1-(1-x_k)e_{ik}\}, \forall~i\neq 1. \]
The sequence of minimizing nodes $ k $ then indicates the MPP from node $ i  $ to any source and can be computed using (\ref{spreading}) with 
\begin{equation}\label{eqpro}
f(\xh_k(t), e_{ik}) = 1 - (1 - \xh_k(t))e_{ik}.
\end{equation}
If $ 0<\sigma \leq e_{ik}<1-\sigma, $ this is progressive and increasing in $ \xh_k(t),  $  though \emph{decreasing in $ e_{ik} $.}
Figure \ref{fig:probability}  illustrates an example execution, with node A (red) as the source, in which each state estimate converges within 5 rounds.


A key shortcoming of (\ref{spreading}), however, is that underestimates can rise very slowly in the presence of small $ e_{ik} $. 
Consider for example (\ref{abf}) with nodes 1 and 2 having the smallest estimates and sharing a short edge $ e. $ At successive instants $ \dh_1(t+1)=\dh_2(t)+e $ and $ \dh_2(t+1)=\dh_1(t)+e $,  i.e. each rises in small increments of $ e $ (and as shown in \cite{TAC}) converge slowly. The generalization below accelerates this slow convergence.

\subsection{A more general spreading block}\label{sgalg}

The most general $ G $-block, given in \cite{viroli2018engineering}, is a generalization of (\ref{genG}) in that  state estimates are updated as
\begin{equation}\label{gengenG}
\xh_i(t+1)=F\left (\tilde{x}_i(t+1), \xh_i(t), v_i\right ) 
\end{equation}
with $\tilde{x}_i(t+1)$ obeying
\begin{equation}\label{genG}
\tilde{x}_i(t+1)=\min\left\lbrace \min_{k\in \cN(i)}\left\lbrace f\left (\xh_k(t), e_{ik} \right ) \right\rbrace, s_i  \right\rbrace, \forall t\geq t_0 ,
\end{equation}
where $v_i$ are  environmental variables and $ f(\cdot,\cdot) $  remains progressive and monotonic. The function $ F(\ell_1,\ell_2,v) $ is \emph{raising},  i.e. for  finite $M\geq 0$, $\delta > 0$ and $D \geq 0,$  obeys
\begin{equation}\label{raising}
F(\ell_1,\ell_2,v) = \begin{cases}
\ell_1 & \ell_2  \geq M  \mbox{ or } |\ell_2-\ell_1 | \leq D\\
g(\ell_2) & \mbox{otherwise}
\end{cases},
\end{equation}
where   $g(x)$ is finite for finite $x$ and obeys
\begin{equation}\label{strg}
g(x) \geq x+\delta.
\end{equation}
Invocation of the second bullet of 
 (\ref{raising}),  speeds the initial ascent of $ \xh_i(t)$, ameliorating the problem of  the slow rise in underestimates experienced by (\ref{spreading}). 
 On the other hand, the first bullet renders (\ref{genG}) identical to (\ref{spreading}). 
 As the second bullet of (\ref{raising})  changes  $ \xh_i(t) $, the stationary point of (\ref{genG}-\ref{strg}) is identical to that of (\ref{stationary}). Thus this  algorithm  spreads the same information as (\ref{spreading}), while  \emph{accelerating the rise of underestimates}. Observe also that  $ D=\infty $ and/or $ M=0 $, reduces (\ref{genG}-\ref{strg}) to (\ref{spreading}).
 
 The version of (\ref{genG}-\ref{strg}) in \cite{viroli2018engineering} sets the \emph{dead zone} variable as $ D=0 $. In  face of persistent structural perturbations in $ e_{ik} $, $ l_2=l_1 $ cannot be sustained. Consequently, regardless of the size of  perturbations, with $ D=0 $, $ \xh_i $ will regularly rise to the limit of the {\em modulation threshold} $ M $, then fall, and then rise again. On the other hand we will show that if $ D $ is sufficiently greater than $ \epsilon $, the bound on the perturbation, then (\ref{genG}-\ref{strg}) will have ultimate bounds proportional to $ \epsilon. $ \emph{This raises an essential trade-off}. Too large a $ D $ slows convergence though imparts greater robustness to perturbations. Such a compromise is inherent to most dynamic systems. Slower convergence improves noise performance.
 
 Another key difference is that \cite{viroli2018engineering}
assumes that $ \xh_i $ belong to a Noetherian ring with $ M $
its \emph{maximal element.} This implicitly assumes that the algorithm is \emph{a priori} bounded. For  distance estimation this means a  prior assumption on the diameter, which is unappealing in the context of open systems.

\begin{figure}[h]
	\centering
	\includegraphics[width = 1\columnwidth]{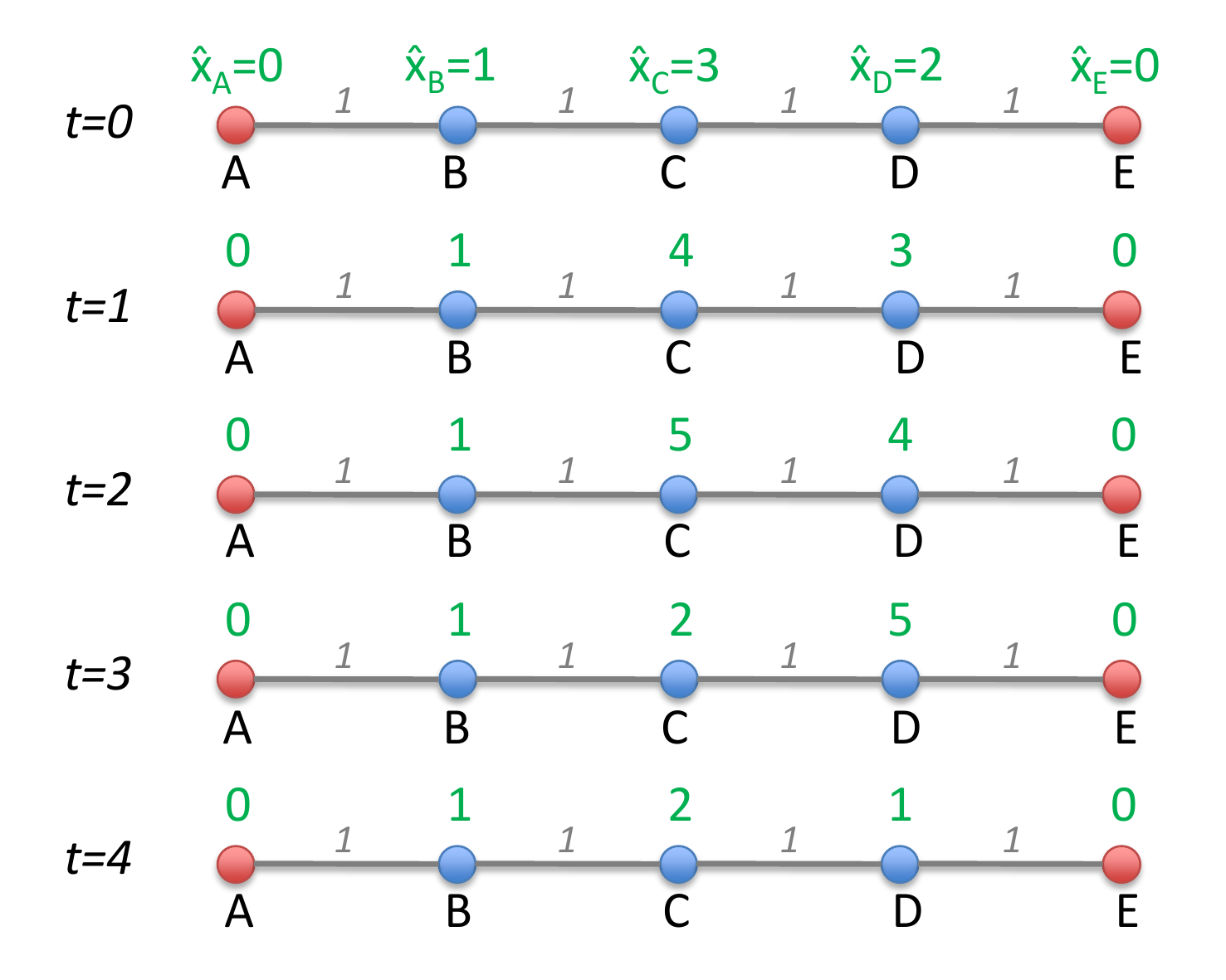}
	\caption{Illustration of GABF and sets $\cA(t), \cE(t), \cR(t)$ and $\cU(t)$. Each edge length in the graph is 1, $M = 4, \delta = 1$, $D = 0$, $s_i = 0$ for $i = A, E$ and $s_i = \infty$ otherwise. In this case, $\cR(1) = \cA(1) = \{A,B,E \}$, $\cU(1) = \cE(1) = \{C,D\}$. }
	\label{fig:sets}
\end{figure}

\todo[inline]{Consider dropping Figure 3, as we don't use it for much, along with the examples that reference it.}

The generalized Adaptive Bellman-Ford algorithm (GABF), presented and analyzed \emph{without proofs} in \cite{mo2019global},  is  a specific example of (\ref{gengenG}), that is an accelerated ABF. In GABF, $f(\xh_k(t), e_{ik})$  follows (\ref{sum}) with $e_{ik}$ the edge length between $i$ and $k$, $\xh_k(t)$ the distance estimate of $k$ at time $t$, $s_i = 0$ if $i$ is a source while $s_i = \infty$ if $i$ is a non-source node. In Figure \ref{fig:sets}, nodes in red are sources, each edge length in the graph is 1, and numbers in green represent the state estimates. Variables $M, D$ and $\delta$ in (\ref{raising}) and (\ref{strg}) are set as $4, 0$ and $1$ respectively.  Convergence occurs in four rounds.

\subsection{Definitions and Assumptions}\label{sdef}
We define $\cS(t)$ as comprising nodes in $\cS^*$ that acquire their maximum values at time $t$,
\begin{equation}\label{st}
\cS(t) = \{i \in \cS^*~|~ \xh_{i}(t) = s_i\},
\end{equation}
and we say $i$ is a source at time $t$ if $i \in \cS(t)$.
The following assumption holds in this paper.
\begin{assumption}\label{amain}
Graph $\cG$ is connected, $\infty>e_{ik}=e_{ki}\geq e_{\min} >0$,   $\cS^*$ defined in (\ref{source}) in nonempty and
\begin{equation}\label{mins}
 s_{\min} = \underset{j \in \cS^*}{\min}\{s_j \} \geq 0, ~ \forall i \in V.
\end{equation}
Further 
\begin{equation}\label{setsmin}
\cS_{\min} = \{i \in V | s_i = s_{\min}  \}.
\end{equation}
\end{assumption}
As in any given iteration the estimated state of a node is obtained by one of the bullets in (\ref{raising}), at each \emph{t}, we partition \emph{V} into two sets defined below.
\begin{definition}\label{dAE}
The set $\cA(t)$ (ABF type nodes) comprises all nodes that use the first case in (\ref{raising}) to obtain $ \xh_i(t) $,  i.e. in (\ref{genG}), $ \xh_i(t)=\xt_i(t) $. Define  the set of extraordinary nodes $\cE(t)=V\setminus \cA(t)$ to be those that use the second case in (\ref{raising}).
\end{definition}
The next definition defines a \emph{(current) constraining node}.
\begin{definition}\label{dcons}
For $i \in \cA(t)$, if $ \xh_i(t)=s_i=x_i $ then $ i $ is its own current constraining, or constraining node at $ t $. Otherwise the minimizing $k\neq i$ in (\ref{genG}) used to find  $ \xh_i(t) $, is $ i $'s  constraining node at  $t$. If  $i \in \cE(t)$, then $ i $ is its own  constraining node at  $t$. The  constraining node of $ i $ at  $t$ is said to constrain $ i $ at $ t. $
\end{definition}

\section{Characterizing Stationary Points}\label{sstat}
This section characterizes the stationary point of (\ref{genG}-\ref{strg}) which as explained in Section \ref{salg} is identical to the stationary point of (\ref{spreading}) given in (\ref{stationary}). Observe that these comprise two sorts of values. Those where $ x_i=s_i. $ Those where $ x_i<s_i. $ We call the former \emph{sources} and their set is defined as 
\begin{equation}\label{sourceset}
\cS_\infty=\{i| x_i=s_i\}.
\end{equation}
Evidently 
\begin{equation}\label{sta}
x_i = \begin{cases}
s_i & i \in \cS_{\infty} \\
\underset{k \in \cN(i)}{\min} \{f(x_k,e_{ik})\} & i \notin \cS_{\infty}
\end{cases}
\end{equation}
As shown by example in Section \ref{scdc}, not all members of $ \cS^* $ are sources. 
To establish the existence of at least one stationary point we  make a definition.
\begin{definition}\label{xij}
	As the graph is connected, 
	there is a path from every node to every other node. Define $ \cP_{ji} $ to be the set of all paths from  $ j $ to $ i $, including $ j=i $.  Denote such a path $ \cP\in \cP_{ji} $, e.g. $ l_0 \rightarrow l_1 \rightarrow \cdots,\rightarrow l_L = i $,  by the ordered set  
	$\cP=\{j = l_0,l_1, \cdots, l_L = i\}$. In particular the path from  $ i $ to $ i $, will be the solitary node:
	\begin{equation}\label{Pii}
	\cP_{ii}=\{\{i\}\}.
	\end{equation}
	Consider  the  recursion,
	\begin{equation}\label{rucom}
	x^*_{l_k}(\cP) = \begin{cases}
	s_{l_k} & k = 0  \\
	f(x^*_{l_{k-1}}(\cP),e_{l_{k-1}l_k})  & k\in \{1,\cdots, L\}
	\end{cases}.
	\end{equation}
	Define $x_{ji}$ as the smallest value $x^*_i(\cP)$ can have among all the paths from $j$ to $i$,  i.e.
	\begin{equation}\label{xji}
	x_{ji}=\min_{ \cP\in \cP_{ji}}\{x^*_i(\cP)\}.
	\end{equation} 
	Further define
	\begin{equation}\label{xb}
	\xb_i=\min_{j\in V}\{x_{ji}\}.
	\end{equation}
\end{definition}
This sequence mimics the second case of (\ref{sta}) sans minimization. The $ \xb_{i} $ are uniquely determined by the structure of the graph and  will provide a characterization of  $ \cS_{\infty} $ and the $ x_i. $  Key points stemming from the fact that $ f(a,b)>a $ are:

\begin{itemize}
	\item[(i)] Only sequences that commence at a node with finite maximum values yield finite $ x_k^* $,  i.e.
	$x _{ji}=\infty $ iff $ j\notin\cS^*. $
	\item[(ii)] One has $ x_{ij}=s_i $ iff $ j=i $ and further
	\begin{equation}\label{barsource}
	\xb_{i}=s_i, \mbox{ iff }  x_{ji}\geq s_i, ~\forall ~j\neq i.
	\end{equation}
	\item[(iii)] Given an $ \xb_{i}\neq s_i $ there is a $ j\neq i $,  $ \cP_i\in\{\cP_{ji}\} $ and $ k\in \cN(i) $, the penultimate node in $ \cP_i $ such that
	\begin{equation}\label{irecur}
	\xb_i=f(x_k^*(P_i),e_{ik}).
	\end{equation}
	Either for this $ k, $ $\xb_k=s_k  $ or there are  $ m $,  $ \cP_k\in\{\cP_{mk}\} $ and $ l\in \cN(k) $, the penultimate node in $ \cP_k $ such that
	\begin{equation}\label{krecur}
	\xb_k=f(x_l^*(P_k),e_{lk}).
	\end{equation}
	\item[(iv)] Because $ f(\cdot,\cdot) $ is progressive, in the sequence (\ref{rucom}), $ x_{l_k}^*(\cP)>x_{l_{k-1}}^*(\cP) $.
\end{itemize}
The next lemma  concerns the scenario in (iii).
\begin{lemma}\label{lnoi}
	Under assumption \ref{amain} consider the quantities defined in (iii) above. Suppose $ \xb_{i}\neq s_i $ and $ \xb_{k}\neq s_k $. Then $ i\notin \cP_k. $
\end{lemma}
\begin{proof}
	To establish a contradiction suppose  $ i\in \cP_k. $ Then because of (\ref{rucom}-\ref{xb}) $ \xb_i\leq x_i^*(P_k) < \xb_k \leq x_k^*(P_i).  $ On the other hand as $ k\in \cP_i $, $ \xb_i > x_k^*(P_i) $ leading to a contradiction.
\end{proof}

We now show that $ \xb_{i} $ obey a recursion like (\ref{stationary}), thus proving that they represent a stationary point.

\begin{lemma}\label{lbrecursion}
	Under Assumption \ref{amain}, then $ \xb_{i} $ in Definition \ref{xij} obeys:
	\begin{equation}\label{brecursion}
	\xb_i=\min\left\lbrace \min_{k\in \cN(i)} f\left (\xb_k, e_{ik} \right ) , s_i  \right\rbrace, \forall ~i\in V .
	\end{equation}
\end{lemma}
\begin{proof}
	From Definition \ref{xij}, in particular, (\ref{xji}) and (\ref{barsource}) and the recursion in (\ref{rucom}), and the fact that $ s_i=x_{ii} $ there holds:
	\begin{flalign}
	\xb_i=&\min\left \lbrace\min_{j\in V\setminus \{i\}}\{x_{ji}\},s_i\right \rbrace\nonumber\\
	=&\min\left \lbrace\min_{j\in V\setminus \{i\}}\left \{ \min_{ \cP\in \cP_{ji}}\{x^*_i(\cP)\}\right\},s_i\right \rbrace\nonumber\\
	=&\min\left \lbrace\min_{j\in V\setminus \{i\}}\left \{ \min_{ \cP\in \cP_{ji}}\left \{\min_{k\in \cN(i)}\{f(x_k^*(\cP),e_{ki})\}\right \}\right\},s_i\right \rbrace.\label{{initkeysequence}}
	\end{flalign}
The nature of the recursion in (\ref{rucom}) ensures  that for every $ k\in\cN(i) $ and $ \cP\in \cP_{ji} $ there is a 
	$ \bar{\cP}\in \cP_{jk} $ such that in (\ref{{initkeysequence}}), the minimizing $ x_k^*(\cP) $ equals the minimizing $ x_k^*(\bar{\cP}) $. Thus,
	\[ \xb_i= \min\left \lbrace\min_{j\in V\setminus \{i\}}\left \{ \min_{ \cP\in \cP_{jk}}\left \{\min_{k\in \cN(i)}\{f(x_k^*(\cP),e_{ki})\}\right \}\right\},s_i\right \rbrace.\]
	As from Lemma \ref{lnoi},  the minimizing path $ \cP_{jk} $ cannot include $ i, $
	\[ \xb_i= \min\left \lbrace\min_{j\in V}\left \{ \min_{ \cP\in \cP_{jk}}\left \{\min_{k\in \cN(i)}\{f(x_k^*(\cP),e_{ki})\}\right \}\right\},s_i\right \rbrace.\]
	As $ f(\cdot,\cdot) $ is monotonically increasing in the first argument, from (\ref{xji}) and (\ref{barsource}), (\ref{brecursion}) is proved by
	\begin{flalign}
	\xb_i=&\min\left \lbrace\min_{k\in \cN(i)}\left \{f\left ( \min_{  j\in V,\cP\in \cP_{jk}}\left \{x_k^*(\cP)\right \},e_{ki}\right )\right\},s_i\right \rbrace\nonumber\\
	=&\min\left \lbrace\min_{k\in \cN(i)}\left \{f\left ( \xb_k,e_{ki}\right )\right\},s_i\right \rbrace.
	\end{flalign}
\end{proof}

Thus we have established the existence of at least one stationary point.  To establish its uniqueness  we make a definition.

 \begin{definition}\label{dtrue}
 	In (\ref{sta}), if $x_i = s_i$, then we say that $i$ is its own true constraining node. Otherwise, any minimizing $k$ in the second bullet of (\ref{sta}) is a true constraining node of $i$. As $i$ may have more than one true constraining node, its set of true constraining nodes is designated as $\cC(i)$.	
 \end{definition}
As $ f(\cdot,\cdot) $ is progressive we have that,
\begin{equation}\label{ineqconst}
x_i>x_k, ~\forall k\in \cC(i) \mbox{ and } i\notin \cS_{\infty}.
\end{equation}
The following lemma catalogs some crucial properties of true constraining nodes and their implications to stationary points.

\begin{lemma}\label{lnonempty}
	Consider $ x=[x_1,\cdots, x_n]^T $ whose elements obey 
	(\ref{stationary}).  Then the following hold under Assumption \ref{amain}. (A)  Consider any sequence of nodes, without loss of generality $ \{1,2,\cdots,l\} $ such that $ i+1\in \cC(i) $   as defined in Definition \ref{dtrue}. Then this sequence is finite and its last element is in $ \cS_{\infty}, $ defined in (\ref{sourceset}). (B)  The set $ \cS_{\infty} $ is nonempty. (C) The set $ \cS_{\min}\subset \cS_{\infty}. $ (D) All $ x_i $ are finite.
\end{lemma}
\begin{proof}
	Due to (\ref{ineqconst}) the chain in (A) cannot have cycles. As there are only $ N $ nodes it must  end, and  the last element $ l $ must be its own true constraining node  i.e.  $l\in \cS_{\infty}.$ This proves (A), and also (B). Without loss of generality suppose  $ s_1=s_{\min} $.  To establish a contradiction, suppose $ 1\notin \cS_{\infty}. $ Then from (A) there is a sequence of nodes starting from 1 and terminating in $ j\in \cS_{\infty} $, such that each is the true constraining node of its predecessor.  Thus from (\ref{ineqconst}) $ x_j=s_j<s_1=s_{\min} $, violating the definition of $ s_{\min}, $ proving (C). 
	To prove (D) consider $ i\neq 1. $  As the graph is connected there is a path from $ 1 $ to $ i $, comprising nodes $ \{1=l_1\rightarrow l_2\rightarrow \cdots l_k=i.\} $ Then from (\ref{sta}) for each $ n\in \{2,\cdots, k\} $ there holds
	\[ x_{l_n}\leq f(x_{l_{n-1}}, e_{l_n,l_{n-1}}). \]
	Due to the fact that $ f(a,b) $ is  finite for finite $ a,b, $ $ x_{l_n} $ is finite if $x_{l_{n-1}}$ is finite. The result follows as $ x_1 $ is finite.
\end{proof}

We make another definition for proving  uniqueness of the stationary point and  convergence of the algorithm.
\begin{figure}[htb]
	\centering
	\includegraphics[width = 0.5\columnwidth]{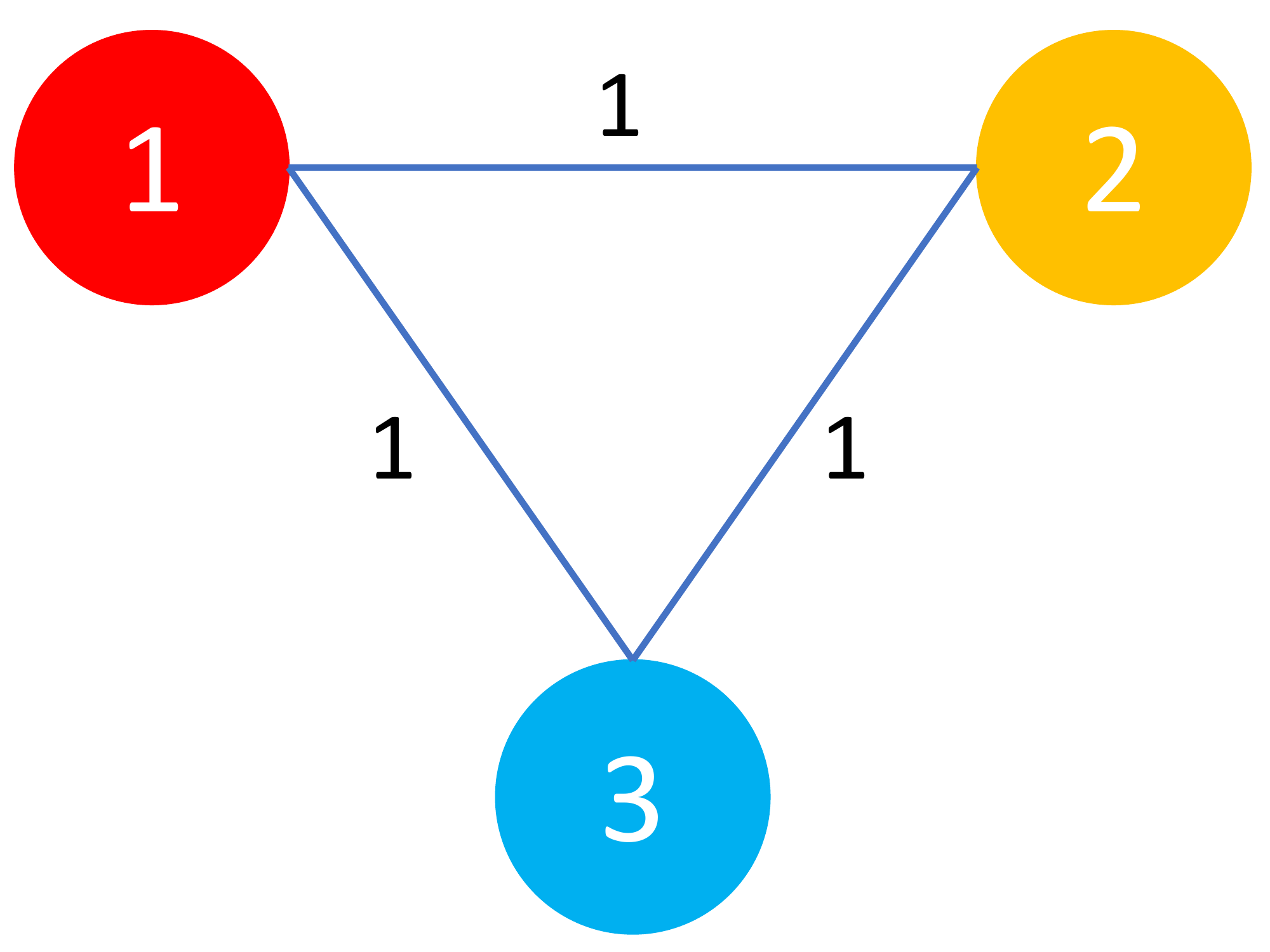}
	\caption{Illustration of graph where $ \cS_{\infty}$ is not a subset of $ \cF_0 $. Here $ s_1=0 $, $ s_2=1 $ and $ s_3=\infty $. All edge lengths are $ 1 $ and $ f(a,b)=a+b. $ In this case $ x_1=0, x_2=1 $ and $ x_3=1. $ Here $ 2\in \cS_\infty $ as $ x_2=s_2. $  However, as $ x_2=x_1+1, $ $ 2\in \cF_1. $}
	\label{fig:3nodes}
\end{figure}

\begin{definition}\label{dF}
	We call a path from a node $i$ to $j \in \cS_{\infty}$ a shortest path, if it starts at $i$, ends with $j \in \cS_{\infty}$, and each node in the path is a true constraining node of its predecessor. We call a shortest path from $i$ the longest shortest path if it has the most nodes among all shortest paths of $i$. The set $\cF_i$ is the set of nodes whose longest shortest paths to the source set have $i + 1$ nodes. We call $\cD(\cG)$ the effective diameter of $\cG$ if the longest shortest path among all $i \in V$ has $\cD(\cG)$ nodes.
\end{definition}
From Lemma \ref{lnonempty}, the effective diameter is always finite.
If a node $i$ has two shortest paths, one with two and the other with three nodes, then $i \notin \cF_1$ but $i \in \cF_2$.
It is tempting to believe that $ \cF_0=\cS_{\infty} $. However,  the scenario of Figure \ref{fig:3nodes} provides a counterexample. In this case $ s_1=0 $, $ s_2=1 $ and $ s_3=\infty $. All edge lengths are $ 1 $ and $ f(a,b)=a+b. $ In this case $ x_1=0, x_2=1 $ and $ x_3=1. $ Here $ 2\in \cS_\infty $ as $ x_2=s_2. $  However, as $ x_2=x_1+1, $ $ 2\in \cF_1. $

The following lemma  exposes a key property of the sets $ \cF_i. $

\begin{lemma}\label{lFF}
	Under the conditions of Lemma \ref{lnonempty}, consider $ \cF_i $ given 	in Definition \ref{dF}.  If for some   $ k\in \{1,\cdots, \cD(\cG)-1\} $,   $ \cF_k $ is nonempty then every node in $ \cF_k $ has a true constraining node in $ \cF_{k-1} $. Further $ \cS_{\min}\subset \cF_0\subset \cS_{\infty}. $
\end{lemma}
\begin{proof}
Consider any $ i\in \cF_k. $   From Definition \ref{dF}, starting from   $ i$  there is a sequence containing $ k+1 $ nodes to a $ j\in \cS_{\infty} $  in which each node is the true constraining node of its predecessor. Suppose the second node in this sequence is $ l $. By definition $ l $ is a true constraining node of $ i. $ Also by definition $ l\in \cF_m, $ where $ m\geq k-1. $ If $ m>k-1 $, then for some $ M>k $, $ i\in \cF_M. $ This contradicts the assumption that $ i\in \cF_k. $  Thus $ l\in \cF_{k-1} $. By definition, every node in $ \cF_0 $ is its own true constraining node as otherwise it will belong to some $ \cF_i $, $ i>0$. Thus from Definition \ref{dtrue}, $  \cF_0\subset \cS_{\infty}. $  

Finally consider $ j\in \cS_{\min}. $ By definition $ s_j=s_{\min}\leq s_i $  for all $ i. $ If $ j\in \cF_k, $ with $ k>0 $, then there is a sequence starting from $ j $ to an $ l\in \cS_{\infty} $, such that each node is the true constraining of its predecessor. Thus from the progressive property of $ f(\cdot,\cdot) $, $ s_{\min}=s_j\geq x_j >s_l, $ establishing a contradiction. Thus $ j\in \cF_0 $ and  $ \cS_{\min}\subset \cF_0. $
\end{proof}
\begin{lemma}\label{lF}
Under the conditions of Lemma \ref{lnonempty}, with $ \cF_i $ defined 	in Definition \ref{dF}, 
\begin{equation}\label{pf}
 \cF_i \neq \emptyset, ~\forall ~  i \in \{0,1,\cdots,\cD(\cG) - 1\},
\end{equation}
and for all $ i \in \{0,1,\cdots,\cD(\cG) - 2\} $ each node in $\cF_{i+1}$ has a true constraining node in $\cF_i$.

\end{lemma}
\begin{proof}

We first show by induction that for each $ k\in \{1,\cdots,\cD(\cG)-1\} $, $ \cF_k $ is nonempty.	
From Definition \ref{dF}, $ \cF_{\cD(\cG)-1}\neq \emptyset $, initiating the induction.  Now suppose for some $  L\in \{1,\cdots, \cD(\cG)-1\} $,  $ \cF_L\neq \emptyset. $ Then from Lemma \ref{lFF} every member of $ \cF_{L} $ has a true constraining node in $ \cF_{L-1} $, making $ \cF_{L-1}\neq \emptyset $.  Further, again from Lemma \ref{lFF} every member of $ \cF_1 $ has a true constraining node in $ \cF_0 $ making the latter nonempty. Then Lemma \ref{lFF} proves the result.
\end{proof}

We can now prove the uniqueness of the stationary point.

\begin{theorem}\label{tstat}
		Under the conditions of Lemma \ref{lnonempty},  $ x_i=\xb_i, $ defined in  Definition \ref{xij}, represents the unique stationary point  obeying (\ref{stationary}). Further the  source set is given by:
		\begin{equation}\label{trsource}
		\cS_{\infty} = \{i\in V|\xb_{i} = s_i \}.
		\end{equation}
\end{theorem}
\begin{proof}
	From Lemma \ref{lbrecursion} $ x_i=\xb_i $  is a stationary point and from (\ref{sourceset}), (\ref{trsource}) is the corresponding source set. 
	
	Call $ \xb=[\xb_1,\cdots,\xb_n]^T $ and consider a potentially different stationary point $ x=[x_1,\cdots,x_n]^T. $ As constraining nodes, source sets and the sets $ \cF_i $ depend on the stationary point, in this proof we will distinguish them with the additional argument of the stationary point, e.g. $ \cC(i,x) $.
	
	We first assert that for all $ i, $ $ x_i\geq\xb_i. $ To establish a contradiction suppose for some $ x_i<\xb_i. $ From Lemma \ref{lnonempty}, there is a $ j\in \cS_{\infty}(x)\subset \cS^* $ and a sequence of nodes
	$ i=l_1\rightarrow\cdots \rightarrow l_L=j $ such that
	\begin{equation}\label{newsequence}
	x_{l_{k+1}}=f\left (x_{l_k}, e_{l_k,l_{k+1}}\right ).
	\end{equation}
	From Definition \ref{xij} this means
	\[ x_{ij}\leq x_i<\xb_i, \]
	violating the definition of $ \xb_i. $  Thus indeed $ x_i\geq\xb_i. $

	As from Lemma \ref{lFF} for all $ j\in\{0,\cdots, \cD(\cG)-1\} $, $ \cF_j(\xb)\neq \emptyset $, we  use induction to show that $ x_i=\xb_i $, for all $ i\in \cF_j(\xb). $ 
	Consider any $ i\in \cF_0(\xb). $ As $\cF_0(\xb) \subset \cS_{\infty}(\xb) $,  
	 $ \xb_i= s_i. $ As by definition,  $ s_i\geq x_i\geq \xb_{i}=s_i $, one must have $ x_i=s_i. $

	To sustain the induction  assume that for some $ 0\leq L<\cD(\cG)-1 $, 
	 $ x_k=\xb_k, $ for all $ k\in \cF_L(\xb). $ To establish a contradiction suppose for some $ i\in \cF_{L+1}(\xb) $, $ x_i\neq \xb_i. $ By Lemma \ref{lF} there is a $ k\in \{\cC(i,\xb)\bigcap\cF_L(\xb)\} $, By the induction hypothesis, $ \xb_k=x_k. $ Then as $ k $ is a neighbor of $ i $, from (\ref{sta})
\begin{eqnarray}
x_i=\min_{l\in \cN(i)}\{f(x_l,e_{il})\}\leq f(x_k,e_{ik}) =f(\xb_k,e_{ik}) =\xb_i.
\end{eqnarray}
As $ x_i\geq \xb_i $, one obtains, $ x_i=\xb_i. $

\end{proof}

Thus we have characterized the stationary point given by (\ref{stationary}) and proved its uniqueness.


\section{Global uniform  asymptotic stability}\label{guas}
Having established the existence and uniqueness of the stationary point $ x_i $, in this section, we prove that  the state estimates $ \xh_i(t) $ yielded  by (\ref{gengenG})-(\ref{strg}), globally, uniformly converge  to these $ x_i $, in graphs   without perturbations  i.e. when  $e_{ik}$ and $s_i$  do not change.
The key steps of the proof are:
\begin{itemize}
\item[(a)] We  show that all underestimates are eventually eliminated,  i.e. for all $ i, $  and sufficiently large $ t, $ $ \xh_i(t)\geq x_i. $  This is done by using the progressive property of $ f(\cdot,\cdot) $  and the second case of (\ref{raising}), which causes $ \xh_{i} $ to increase.
\item[(b)] We show that once underestimates are eliminated, the moment a source node $ i\in \cS_{\infty}\supset\cF_0 $ invokes the first bullet of (\ref{raising}),  i.e. lies in $ \cA(t) $, it converges to $ s_i. $ Similarly if at a time $ t' $ and beyond,  $ \cF_0$  to $\cF_{L} $ defined in Definition \ref{dF}  have converged, if  $ i\in  \cF_{L+1} $ invokes the first case of (\ref{raising}) then $ \xh_i(t) $ converges to $ x_i. $
\item[(c)] We then show that over every finite interval, the first bullet of (\ref{raising}) must be invoked at least once as each invocation of the second case of (\ref{raising}) by $ i  $ increases $ \xh_i $ by $ \delta, $. Thus   the repeated use of the second bullet by $ i $ will eventually induce
$ \xh_i>M $, forcing $ i $ to use the first bullet of (\ref{raising}).

\item[(d)] To show (a) we define two time varying sets that partition $ V $. The first, $ \cR(t), $ the set of nodes \emph{rooted} to sources,  comprises elements of $ \cS(t) $, the  source set at time $ t $,  and all nodes constrained by elements of $ \cR(t-1). $ We show that for all  $i\in \cR(t) $, $ \xh_i(t)\geq x_i. $  Nodes in the second \emph{unrooted} set $ \cU(t) $ must also eventually have no  state estimates that are underestimates.


\end{itemize}

\subsection{Key lemmas underlying (b,c)}\label{sbc}
This section is dedicated to key lemmas that underlie  (b) and (c).
The first lemma provides and proves a key mechanism behind (b). Specifically, should after a time $ t_1 $  no neighbor of a node $ i $ ever carry underestimates and its true constraining node converges, then for all $ t> t_1 $    $ \xt_i(t) $ in (\ref{genG})  equals $ x_i $.


\begin{lemma}\label{lgenG}
	Consider (\ref{genG}) and (\ref{raising}) under Assumption \ref{amain} and a node $ i $ and a time $ t_1 $ such that the following hold for all $ t\geq t_1. $  If $ i\in \cS_{\infty} $, $  \xh_j(t)\geq x_j $ for all  $ j\in \cN(i) $.  If $ i\notin \cS_{\infty} $, (i) m $  \xh_j(t_1)\geq x_j $ for all $ j\in \cN(i) $; and   (ii) with $ k $ a true constraining node of $ i , $ $ \xh_k(t) =x_k$.  
	Then for all $ t> t_1, $  $ \xt_i(t) $  defined in (\ref{genG}) equals $ x_i. $
\end{lemma}
\begin{proof}
	Suppose $ i\in \cS_{\infty} $, then $ i $ is its own true constraining node. Then from Definition \ref{dcons} and (\ref{stationary}),   and the fact that $ f(a,b) $ is strictly increasing in $ a $, from (\ref{genG}) there holds for all $ t\geq t_1 $
	\begin{align*}
	\xt_i(t+1)&=\min\left \{\min_{j\in \cN(i)}\{f(\xh_j(t),e_{ij})\},s_i  \right \}\\
	&=\min\left \{\min_{j\in \cN(i)}\{f(x_j,e_{ij})\},s_i  \right \}=s_i=x_i.
	\end{align*}
	If $i \notin \cS_{\infty} $ then from Definition \ref{dcons} and (ii)
	\begin{align*}
	x_i&=\min_{j\in \cN(i)}\{f(x_j,e_{ij})\}=f(x_k,e_{ik})\\
	&=f(\xh_k(t),e_{ik}), ~\forall ~t\geq t_1.
	\end{align*}
	Further  from (i)
	\[ \min_{j\in \cN(i)}\{f(x_j,e_{ij})\}\leq \min_{j\in \cN(i)}\{f(\xh_j(t),e_{ij})\} ~\forall ~ t\geq t_1. \]
	As $ k\in \cN(i) $ one thus has that
	\begin{align*}
\min_{j\in \cN(i)}\{f(\xh_j(t),e_{ij})\}=f(x_k,e_{ik}), =x_i ~ t\geq t_1.
\end{align*}
	By definition  $ i\notin \cS_{\infty} $ implies that $ x_i< s_i. $ Thus from (\ref{genG})   for all $ t\geq t_1, $ $ \xt_i(t+1) $ equals
	\begin{align*}
	\min\left \{\min_{j\in \cN(i)}\{f(\xh_j(t),e_{ij})\},s_i  \right \}=\min\left \{x_i,s_i  \right \}=x_i.
	\end{align*}

\end{proof}

In view of Definition \ref{dAE} under the conditions of Lemma \ref{lgenG}, if at any $ t>t_1 $, $ i\in \cA(t)\bigcap \cS_{\infty} $, then $\xh_i(t)= \xt_i(t)=s_i. $ The next lemma shows that if after $ t_1 $ defined in Lemma \ref{lgenG}, $ i $ ever enters $ \cA(t) $ then $ \xh_i(t) $ converges immediately to $ x_i. $
\begin{lemma}\label{lArecur}
	Consider (\ref{gengenG}-\ref{strg}). Suppose the conditions of Lemma \ref{lgenG} hold, and for some $  t_2>t_1 $, $ i\in \cA(t_2) $ defined in Definition \ref{dAE}. Then for all $ t\geq t_2 $, $ \xh_i(t) =x_i.$
\end{lemma}
	\begin{proof}
In view of Lemma \ref{lgenG}, we need to show that for all  $ t\geq t_2 $, $ i\in \cA(t). $
Use induction. By  hypothesis, $ i\in \cA(t_2). $ Now suppose for some $ t\geq t_2>t_1 $ , $ i\in \cA(t). $ Then from Definition \ref{dAE} and Lemma \ref{lgenG}, $ \xh_i(t)=\xt_i(t)=x_i. $ Further, also from  Lemma \ref{lgenG}, $ \xt_i(t+1)=x_i=\xh_i(t). $ Thus from (\ref{genG}) and the first bullet of (\ref{raising}), $ \xh_i(t+1) = \xt_i(t+1), $   i.e. $ i\in\cA(t+1). $

\end{proof}
Thus under the conditions of Lemma  \ref{lgenG} all it takes for $ \xh_i(t) $ to converge after $ t_1 $ is for $ i $ to enter the ABF set.
The next lemma  bounds the time, described in (c), for this to  happen.

\begin{lemma}\label{lconverge}
	Under the conditions of Lemma \ref{lgenG}, consider (\ref{gengenG})-(\ref{strg}). Suppose $ \xh_i(t_2)=m_i $ for some $t_2 > t_1$.
	Define
	\[ t_3=t_2 + 1 + \min\left \{\left \lceil \frac{M-m_i}{\delta} \right \rceil, 0\right \} .\]
	Then for all $ t>t_3 $, $ \xh_i(t)=x_i. $
\end{lemma}
\begin{proof}
	Suppose $ \xh_i(t_2)\neq x_i. $ Then from Lemma \ref{lArecur}, $ i\in \cE(t_2) $. Now suppose $i \in \cE(t)$ for all $t_2 \leq t \leq t'$. In this case from the second bullet of (\ref{raising}), $t' \leq t_3. $ From the first bullet of (\ref{strg}) this means $ i\in \cA(t_3+1). $ Then Lemma \ref{lArecur} proves the result.
\end{proof}

So if no neighbor of a non-source $ i $ carries an underestimate and  at least one of its true constraining nodes  has converged, then $ \xh_i(t) $  converges the moment $ i $ enters $ \cA(t), $ which it must in a time quantified in Lemma \ref{lconverge}. The same is true if the states of all neighbors of a source $ i $ have exceeded $ s_i $, albeit under a weaker condition. 
The next subsection proves a key property that facilitates this convergence: the eventual removal of underestimates noted in (a) at the start of this section.

\subsection{Disappearance of underestimates}\label{sad}
We first define the two time varying sets $\cU(t)$ and $\cR(t)$ mentioned in (d)  at the beginning of the section. 
\begin{definition}\label{dur}
	The set of nodes rooted to sources is $\cR(t+1) = \cS(t+1)\bigcup P(t+1)$ with $\cS(t+1)$ as in (\ref{st}) and $P(t+1)$ comprising those whose  constraining node at  $t+1$ is in $\cR(t)$. Further $\cR(t_0) = \cS(t_0)$. The  unrooted set  is  $\cU(t) = V\setminus \cS(t)$.
\end{definition}  
Evidently, $ \cU(t+1) \bigcap \cS(t+1)=\emptyset. $ 
As every node must have a  constraining node at every $ t $, and members of $ \cR(t+1) $ are either in $ \cS(t+1) $ or are constrained by members of $ \cR(t) $, each member of $ \cU(t+1) $ must be  constrained at time $ t+1 $ by one of $ \cU(t) $.  
Thus 
\begin{equation}\label{implies}
\cU(t) = \emptyset \implies \cU(t+1) = \emptyset.
\end{equation}
Sets  in definitions \ref{dAE} and  \ref{dur} are exemplified through GABF in Figure \ref{fig:sets}. In this case, $\cS(0) = \{A, E\}$ as $\xh_i(0) = 0 = s_i$ for $i = A$ or $E$. At $t = 1$, $\tilde{x}_B(1) = \xh_A(0) + e_{AB} = \xh_B(0) = 1$, as $D = 0$, node $B$ will take $A$ as the current constraining node and use the first bullet of (\ref{raising}) to update its estimate, leading to $B \in \cA(1) \cap \cR(1)$. Meanwhile, as $\tilde{x}_D(1) = \xh_E(0) + e_{DE} = 1 \neq \xh_4(0)$, $D = 0$ and $\xh_D(0) < M$, node $D$ will update its estimate using the second bullet of (\ref{raising}) and take itself as the current constraining node, then $D \in \cE(1)\cap\cU(1)$.

We will now show that underestimates in $ \cU(t) $ must eventually disappear. To this end define
\begin{equation}\label{minU}
\xh_{\min}(t) = \underset{j \in \cU(t)}{\min}~\{\xh_j(t) \}~ \mathrm{if}~ \cU(t) \neq \emptyset,
\end{equation}
Define:
\begin{equation}\label{xmax}
x_{\max} = \underset{k \in V}{\max}~ \{x_k\},
\end{equation}
and
\begin{equation}\label{tstar}
T = \bigg\lceil \frac{x_{\max} - \xh_{\min}(t_0)}{\min\{\sigma,\delta \}} \bigg\rceil.
\end{equation}
\begin{lemma}\label{luin}
	Consider (\ref{gengenG})-(\ref{raising}) under Assumption \ref{amain}, with  $\cU(t)$, $\xh_{\min}(t)$, $ x_{\max} $ and $ T $ defined in Definition \ref{dAE},  (\ref{minU}), (\ref{xmax}) and (\ref{tstar}), respectively. Then  (\ref{uin}) and (\ref{aft}) hold   while  $\cU(t) \neq \emptyset$:
	\begin{equation}\label{uin}
	\xh_i(t) \geq \xh_{\min}(t_0) + \min\{\sigma,\delta\}(t - t_0),~ \forall i \in \cU(t)
	\end{equation}
	\begin{equation}\label{aft}
	\xh_{i}(t) \geq x_{\max} \geq x_i,~ \forall i \in \cU(t)\mbox{ and }\forall t \geq t_0 + T.
	\end{equation}
\end{lemma}
\begin{proof}
	Because of (\ref{implies}), $ \cU(t) $ is nonempty only on a single contiguous time interval commencing at $ t_0 $.
	Thus, from (\ref{xmax}) and (\ref{tstar}), (\ref{aft}) will hold if (\ref{uin}) holds.
	
	We prove (\ref{uin}) by induction in $ t\geq t_0. $ It clearly holds for $t = t_0$. Thus suppose it holds at some $ t\geq t_0. $    If $ \cU(t+1) $ is empty then it remains so for all future values. So assume $\cU(t+1)\neq \emptyset$  i.e.  $\cU(t)\neq \emptyset$.  Suppose $ i\in \cU(t+1) $ is such that
	 $\xh_i(t + 1) = \xh_{\min}(t + 1)$. From the remark after Definition \ref{dur}, $j$ the current constraining node of $i $ is in $\cU(t)$. Suppose $ i\in \cE(t+1) $ defined in Definition \ref{dAE}, then from Definition \ref{dcons}, $ j=i $.    The induction hypothesis and (\ref{strg}) yield:
	\begin{eqnarray}
	\xh_{i}(t + 1) & = & \xh_{\min}(t + 1) \nonumber \\
	& \geq & \xh_{j}(t) + \delta \geq \xh_{\min}(t)+\min\{\sigma,\delta\} \label{eq:ude} \\
	&\geq& \xh_{\min}(t_0) + \min\{\sigma,\delta\}(t + 1 - t_0). \nonumber
	\end{eqnarray}
	If $i \in \cA(t + 1)$, then, $ i\notin \cS(t+1) $,  i.e. $ \xh_i(t+1)\neq s_i. $ Thus,
	\begin{align*}
	\xh_{i}(t + 1)  = & \xh_{\min}(t + 1) = f(\xh_{j}(t),e_{ij}) \\
	 \geq & \xh_{j}(t) + \sigma 
	\geq \min\{\sigma,\delta\}(t + 1 - t_0). \nonumber
	\end{align*}
\end{proof}
We now show that after $ t_0+T $ \emph{all}  $ \xh_{i}(t) $ are overestimates.

\begin{lemma}\label{lover}
	Under the conditions of Lemma \ref{luin}, 
	\begin{equation}\label{all}
	\xh_i(t) \geq x_i, ~\forall~ i\in V, \mbox{ and } t\geq T+t_0.
	\end{equation}
\end{lemma}
\begin{proof}
	We will first show by induction that whenever $ \cR(t) $ given in Definition \ref{dur} is nonempty, $ \xh_i(t)\geq x_i $ for all $ i\in \cR(t). $ Then as $ \cU(t)=V\setminus \cR(t) $, the result will follow from Lemma \ref{luin}.  If $ \cR(t')\neq \emptyset  $, then there is a $t_0\leq  t_4 \leq t'$  such that $ \cR(t)\neq \emptyset $, for all $ t_4\leq t\leq t' $ and $ \cR(t_4)=\cS(t_4). $ Clearly by definition of $ \cS(t)  $, $ \xh_i(t_4)=s_i \geq x_i, $ for all $ i\in \cS(t_4)=\cR(t_4) $. If $ t'=t_4, $ then all elements of $ \cR(t') $ carry overestimates.
	
	If $ t'>t_4 $ then use induction on $ t_4\leq t\leq t'. $  Suppose $ x_i\leq \xh_i(t) $ for some $ t_4\leq t<t', $ and  all $ i\in \cR(t). $ Consider any $ i\in \cR(t+1) $. Then from Definition \ref{dur}, either $ i\in \cS(t+1) $, in which case the result holds,  or $ j  $ the current constraining node of $ i$ is in $\cR(t). $  Then  by the induction hypothesis, $ \xh_j(t)\geq x_j. $ As $ \xh_i(t+1)\neq s_i $,  and $ f(a,b) $  is increasing in $ a, $
if $ i\in \cA(t+1) $,
there follows:
\begin{align*}
\xh_i(t+1)&=f(\xh_j(t), e_{ij})\geq f(x_j, e_{ij})\geq x_i.
\end{align*}
If $ i\in\cE(t+1) $ then it is its own true constraining node and $i\in \cR(t).   $ Thus by the induction hypothesis, $ \xh_i(t)\geq x_i. $ Thus from (\ref{strg})
 $ \xh_i(t+1)\geq \xh_i(t)+\delta>x_i. $

\end{proof}

Thus we have  established (a)  described at the beginning of this section. In the next section we prove GUAS.

\subsection{Proof of convergence}\label{scomplete}

Define the smallest stationary value in $\cF_i$ as
\begin{equation}\label{ximin}
x_{i\min} = \underset{j \in \cF_i}{\min}~ \{x_j\}.
\end{equation}
From  Lemma \ref{lF}, we have $x_{0\min} = s_{\min}$. Define a sequence 
\begin{equation}\label{ti}
T_i = \max \bigg\{0,\bigg\lceil \frac{M - x_{i\min}}{\delta} \bigg\rceil\bigg\} + 2.
\end{equation}

We then have the main result of this section, proving the convergence of each $ \xh_i(t) $ to $ x_i. $   Specifically, we will show by induction that with $ T $ defined in (\ref{tstar}),   for  all $ i\in \{0,\cdots,\cD(\cG)-1\} $,  the elements  of $ \cF_0, \cdots, \cF_i$, defined in Definition \ref{dF}, converge by the time $ T+\sum_{j=0}^{i}T_j. $


\begin{theorem}\label{ttime}
	Consider (\ref{gengenG}) - (\ref{raising}) under Assumption \ref{amain}, with $T_i$ and $T$  defined in (\ref{ti}) and (\ref{tstar}), respectively. Then $ \forall $ $i \in V$, 
	\begin{equation}
	\xh_{i}(t) = x_i,~ \forall t > t_0 + T  + \sum_{i = 0}^{\cD(\cG)-1}T_i.
	\end{equation}
\end{theorem}
\begin{proof}
	We will prove by induction that for every $ L\in \{0,\cdots,\cD(\cG)-1\}, $ 
	\begin{equation}\label{iduction}
	\xh_{i}(t) = x_i , ~\forall~ t\geq t_0+T+\sum_{j = 0}^{L}T_j \mbox{ and } i\in \bigcup_{j= 0}^{L}\cF_j
	\end{equation} 
	Then the result will follow as the  $ \cF_j $ partition $ V. $
	
	Consider $ i \in\cF_0$ and $ t>t_0+T+T_0. $ From Lemma \ref{lFF}, $ i\in \cS_{\infty}. $ As from Lemma \ref{lover}, $ \xh_k(t)\geq x_k  $, for all $ t>t_0+T $ and $ k\in V $, $ i $ satisfies the conditions of Lemma \ref{lgenG} and thus of Lemma \ref{lconverge}. Thus, from Lemma \ref{lconverge}, (\ref{iduction}) holds for $ L=0 $. 
	
	Suppose (\ref{iduction}) holds for some $ N\in \{0,\cdots, \cD(\cG)-2\} $. Consider $ i\in \cF_{N+1}. $ By Lemma \ref{lF}, $ i $ has a true constraining node   $k\in \cF_N. $ By the induction hypothesis $ \xh_k(t)=x_k $ for all $ t>t_0+T+\sum_{j = 0}^{N}T_j $, and $ \xh_l(t)\geq x_l $, for all $ l\in V $. Thus from Lemma \ref{lover}, this $ i $ satisfies the conditions of Lemma \ref{lgenG} and thus of Lemma \ref{lconverge}. Thus, from Lemma \ref{lconverge}, (\ref{iduction}) holds for $ L=N+1 $, completing the proof.
\end{proof}

In fact one can show that this theorem also holds with 
\begin{equation}\label{ts0}
T_0 = \max \bigg\{0,\bigg\lceil \frac{M - \min\{\delta + s_{\min}, x_{\max} \}}{\delta} \bigg\rceil\bigg\} + 2,
\end{equation}
as $ \xh_i(t_0 + T)\geq \min\{\delta + s_{\min}, x_{\max} \}$ for all $i \in \cF_0$.
The fact that the time elapsed between the initial time $ t_0 $ and the time to converge is independent of $ t_0 $ proves  GUAS.

\section{Robustness under perturbations}\label{robust}
In this section, we prove that  (\ref{gengenG}) is ultimately bounded under persistent perturbations in the $ e_{ij} $, albeit with some additional assumptions. 
In particular, the dead zone parameter $ D $ must exceed a value proportional to the magnitude of the perturbation. Otherwise, with probability one $ \xh_i(t) $ will persistently rise to $ M. $ This value is provided in this section. Proofs of this section are in the appendix.

 The first  additional assumption extends the monotonicity property to the second argument of $f(\cdot,\cdot)$. Given that this argument represents edge lengths in most applications, this is a reasonable assumption. As is also standard in most robustness analysis, we also impose a Lipschitz condition.
\begin{assumption}\label{sec}
The function $f(\cdot,\cdot)$ is monotonically increasing with respect to its second argument,  i.e. $f(a,b)$ obeys
\begin{equation}\label{monotonic1}
f(a,b_1) \geq f(a,b_2), ~\mbox{if}~ b_1\geq b_2.
\end{equation}
Further, there exist $ L_i>0 $, such that
\begin{equation}
|f(a,b_1) - f(a,b_2)| \leq L_1|b_1 - b_2|    \label{ul1}
\end{equation}
and
\begin{equation}\label{aa}
|f(a_1,b) - f(a_2,b)| \leq L_2|a_1 - a_2|
\end{equation}
\end{assumption}
The perturbations on $e_{ij}$ are modeled as,
\begin{equation}\label{vare}
e_{ij}(t) = e_{ij} + \epsilon_{ij}(t)
\end{equation}
with 
\begin{equation}\label{ep}
|\epsilon_{ij}(t)| \leq \epsilon < e_{\min},
\end{equation}
where $e_{\min}$ is defined in Assumption \ref{amain}. Notice that the perturbations need not be symmetric,  i.e. we permit
\begin{equation}\label{nonsy}
\epsilon_{ij}(t) \neq \epsilon_{ji}(t). 
\end{equation}
Such perturbations could reflect noise, localization error, or (if coherent) movement of devices. 
In this case, (\ref{genG}) becomes 
\begin{equation}\label{genGt}
\tilde{x}_i(t+1)=\min\left\lbrace \min_{k\in \cN(i)}\left\lbrace f\left (\xh_k(t), e_{ik}(t) \right ) \right\rbrace, s_i  \right\rbrace.
\end{equation}
We  define a \emph{shrunken graph}, for deriving    bounds on underestimates.  It corresponds to the smallest possible values of $e_{ij}$. 
\begin{definition}\label{shrunken}
Given a graph $\cG$, define $\cG^-$ as a shrunken version of $\cG$ such that, $\forall i\in V$ and $j \in \cN(i)$ in $\cG$, $e_{ij}$ becomes $e_{ij}^-$ in $\cG^-$: With $\epsilon$ defined in (\ref{ep})
\begin{equation}\label{shru}
e_{ij}^- = e_{ij} - \epsilon.
\end{equation}
Also consider (\ref{gengenG}) implemented on this shrunken graph,  i.e.
\begin{equation}\label{algX}
\hat{X}_i(t+1)= F( \tilde{X}_i(t+1), \hat{X}_i(t), v_i ), ~\hat{X}_i(0)\leq \xh_i(0).
\end{equation} 
with $\tilde{X}_i(t+1)$ obeying
\begin{equation}
\tilde{X}_i(t+1) = \min\left\lbrace \min_{k\in \cN(i)}\left\lbrace f\left (\hat{X}_k(t), e_{ik}^- \right ) \right\rbrace, s_i  \right\rbrace
\end{equation} 
As $\cG^-$ satisfies the same assumptions as $\cG$ and is perturbation free, we define $X = [X_1,\cdots,X_N]$ as the unique stationary point in $\cG^-$ to which (\ref{algX}) converges.
Further, $\cS_{\infty}^-$ and $\cD(\cG^-)$ denote the source set and the effective diameter of $\cG^-$.
\end{definition}
The unique stationary point in $\cG^-$ obeys
\begin{equation}
\label{recm}
X_i=\min\left\lbrace \min_{k\in \cN(i)} f\left (X_k, e_{ik}^- \right ) , s_i  \right\rbrace, \forall ~i\in V .
\end{equation}
Specifically, from (\ref{sourceset}) and (\ref{sta}), the source set in $\cG^-$ obeys
\begin{equation}\label{trsourcet}
\cS_{\infty}^- = \{i | X_i = s_i   \},
\end{equation}
and the stationary point obeys: 
\begin{equation}\label{stat}
X_i = \begin{cases}
s_i & i \in \cS_{\infty}^- \\
\underset{k \in \cN(i)}{\min}~ \{f(X_k,e_{ik}^-) \} & i \notin \cS_{\infty}^-.
\end{cases}
\end{equation}
Evidently, the following holds in $\cG^-$:
\begin{equation}\label{lem}
s_{\min} \leq X_i \leq s_i, ~ \forall i \in V.
\end{equation}
Define 
\begin{equation}\label{Xmax}
X_{\max} = \underset{k \in V}{\max}~ \{X_k\},
\end{equation}
and 
\begin{equation}\label{tm}
T^- = \bigg\lceil \frac{X_{\max} - \xh_{\min}(t_0)}{\min\{\delta,\sigma\}} \bigg\rceil
\end{equation}
The next lemma shows that the lower bound in $\cU(t)$ will exceed $X_{\max}$ after $t_0 + T^-$.
\begin{lemma}\label{sluin}
Consider (\ref{gengenG}), (\ref{raising}) and (\ref{genGt}), with $\cU(t)$, $\xh_{\min}(t)$, $X_{\max}$ and $T^-$ defined in Definition \ref{dur}, (\ref{minU}), (\ref{Xmax}) and (\ref{tm}), respectively. Then (\ref{suin}) and (\ref{exX}) hold while the set $\cU(t) \neq \emptyset$:
\begin{equation}\label{suin}
\xh_i(t) \geq \xh_{\min}(t_0) + \min\{\sigma,\delta\}(t - t_0),~ \forall i \in \cU(t)
\end{equation}
\begin{equation}\label{exX}
\xh_i(t) \geq X_{\max} \geq X_i,~\forall i \in \cU(t) \mbox{ and } \forall t \geq t_0 + T^-
\end{equation}
\end{lemma}


We now turn to $\cR(t)$ and prove that under perturbations all estimates in $\cR(t)$ are lower bounded by their corresponding stationary values in $\cG^-$. 
\begin{lemma}\label{tover}
Consider (\ref{gengenG}), (\ref{raising}) and (\ref{genGt}), with $\cA(t)$, $\cE(t)$, $\cR(t)$ and $X_i$ as in definitions \ref{dAE}, \ref{dur} and \ref{shrunken}, respectively. There holds:
\begin{equation}\label{t{overestimate}}
\xh_i(t) \geq X_i,~ \forall i \in \cR(t).
\end{equation}
\end{lemma}
Consequently, with Lemma \ref{sluin} and Lemma \ref{tover}, there holds:
\begin{equation}\label{tall}
\xh_{i}(t) \geq X_i,~ \forall i \in V,~ \forall t \geq t_0 + T^-.
\end{equation}
By quantifying the relation between the stationary point in $\cG^-$ and that in $\cG$ in the next lemma, we will show that after $t_0 + T^-$ all estimates are lower bounded.

To this end we define the following function:
\begin{equation}\label{W}
W(L_2,D) = \sum_{i = 0}^{D-1}L_2^i = \begin{cases}
\frac{L_2^D - 1}{L_2 - 1} & L_2 \neq 1 \\
D & L_2 = 1
\end{cases},
\end{equation}
and the summation is zero if the lower limit exceeds the upper.
\begin{lemma}\label{tlbound}
Under Assumption \ref{sec}, with $\epsilon$, $W(\cdot)$ and $\cD(\cG^-)$ defined in (\ref{ep}), (\ref{W}) and Definition \ref{shrunken}, respectively. Then for all $i \in V$ and $t \geq T^- + t_0$, there holds:
\begin{equation}\label{resta}
x_i \leq X_i + W(L_2,\cD(\cG^-)-1)L_1\epsilon.
\end{equation}
\end{lemma}


With $\cF_i$ defined in Definition \ref{dF}, define $X_{i\min}$ as 
\begin{equation}\label{Ximin}
X_{i\min} = \underset{j \in \cF_i}{\min}~ \{X_{j}\}
\end{equation}
Note that $\cS_{\min}$ is a subset of $\cS_{\infty}^-$ as well as $\cF_0$, thus $X_{0\min} = s_{\min}$. Define a sequence 
\begin{equation}\label{tmi}
T_i^- = \max\left \{0,\bigg\lceil \frac{M - X_{i\min}}{\delta}\bigg\rceil   \right \} + 2.
\end{equation}
Then we have the following lemma that recognizes that
to behave acceptable under perturbations the dead zone $ D $ in (\ref{raising}) must be sufficiently large.
\begin{lemma}\label{tpre}
Consider (\ref{gengenG}), (\ref{raising}) and (\ref{genGt}) under Assumption \ref{sec}, with $\epsilon$,  $W(\cdot)$ defined in (\ref{ep}), (\ref{W}), respectively. Suppose $D$ in (\ref{raising}) obeys
\begin{equation}\label{D}
D \geq (W(L_2,\cD(\cG^-) - 1) + W(L_2,\cD(\cG) - 1))L_1\epsilon
\end{equation}
and at a time $ t' \geq t_0 + T^-$ defined in (\ref{tm}), for some  $L \in \{0,1,\cdots,\cD(\cG)-2\}$
\begin{equation}\label{pre}
\xh_{i}(t) \leq x_i + W(L_2,L)L_1\epsilon,~ \forall i \in \cF_L
, ~ \forall t \geq  t'.
\end{equation}
Then with $T_i^-$ define in (\ref{tmi}), there holds:
\[
\xh_{i}(t) \leq x_i + W(L_2,L + 1)L_1\epsilon, ~\forall ~i \in \cF_{L+1}, ~t \geq  t' + T_{L+1}^-.
\]
\end{lemma}

The next theorem proves that the algorithm is ultimately bounded under bounded persistent perturbations and provides an upper bound on the time to attain the ultimate bound.
\begin{theorem}\label{bound}
Under the conditions of Lemma \ref{tpre},
 for all $i \in V$ and $t \geq t_0 + T^- + \sum_{i = 0}^{\cD(\cG)-1}T^-_i$,
\[|\xh_{i}(t) - x_i| \leq \epsilon L_1 \max\left \{W(L_2,\cD(\cG)-1), W(L_2,\cD(\cG^-)-1)\right \}.\]
\end{theorem}

This is a classical ultimate bound  with the bound proportional to $ \epsilon $, the magnitude of the disturbance.
Define
\begin{equation}\label{t0}
T_0^- = \max \bigg\{0,\bigg\lceil \frac{M - \min\{\delta + s_{\min}, X_{\max} \}}{\delta} \bigg\rceil\bigg\} + 2.
\end{equation}
One can show that Theorem \ref{bound} holds for a tighter time bound if one uses $T_0^-$ defined in (\ref{t0}). This is so as one can prove that $\xh_i(t_0 + T^-) \geq \min\{\delta + s_{\min}, X_{\max} \}$ for all $i \in \cF_0$. 

\section{Design choices and discussion}\label{sdesign}
Theorem \ref{bound}  verifies the intuitively clear requirement that the dead zone $ D $ should grow proportionally to the disturbance bound $ \epsilon. $ 
However, as this is a worst-case analysis, it masks the full effects of parameters $ M $, $ \delta $, and $ D. $ 
Looking beyond worst-case analysis, however, we can find that choosing these parameters involves tradeoffs between the convergence speed of underestimates and overestimates.

The convergence of underestimates is upper bounded by $ T $ in (\ref{tstar}), which is in turn  determined by (\ref{uin}), and thus conservatively by the smaller of $ \sigma $ and $ \delta $.
In practice, if $ \sigma  $ is small and the first bullet of raising is invoked too often then underestimates rise slowly,  i.e. the rising value problem will persist.
If the second bullet of (\ref{raising}) is invoked at most times and $ \delta \gg \sigma $   then underestimates decline fast. 
Large  $ D $ or  small $ M $  makes this less likely and  slows convergence, while a large $ \delta \geq M $  speeds convergence by reducing $ T $ and $ T_i. $

For the convergence of overestimates, $ T_i $ gives the worst case time to invoke the first case of (\ref{raising}), whereupon all elements in $ \cF_i $ converge forthwith. 
The worst case analysis  quantifies $ T_i $  by how long it takes for $ \xh_i(t) $ to exceed $ M $ and assumes that the second clause of (\ref{raising}) is invoked until this happens. 
With a large $ D $, however, this time shortens as the first bullet is likely to be invoked more quickly.  

In most cases, the need to alleviate the rising value problem is more compelling as overestimates in algorithms like plain ABF converge in at most $ \cD(\cG)-1 $ steps. 
Accordingly, the desirability of a smaller dead zone $ D $ competes with the requirement of resilience to persistent perturbations as quantified by  (\ref{D}). This of course is common to most dynamical systems where faster convergence generally comes at the price of reduced resilience. We note, however, the following appealing fact: \emph{both the ultimate bound and the required $ D $ are determined exclusively by the perturbation magnitude $ \epsilon $ and the effective diameters of the original and shrunken graph.}

Complementarily, note that in the special case of the algorithm in \cite{mo2018robust}, we effectively have $ M=0 $ and $ D=\infty$.  In this case the second bullet of (\ref{raising}) is never invoked.
Accordingly, a small $ \sigma $ leads to large $ T $ and $ T^- $  and the rising value problem. 
In particular, the algorithm remains GUAS with the same ultimate bound as (\ref{D}) is automatically satisfied. Overestimates however, converge quickly as  $ T_i=T_i^-=2.  $

\section{Simulations}\label{ssim}

In this section, we empirically confirm the results presented in the prior sections through simulations.
We first investigate the effect of parameters in the general spreading block configured as GABF, then compare with the performance of ABF in the presence of persistent perturbations. 
Finally, we illustrate the applicability of the general spreading block to more complex cases with an example of a non-Euclidean distance metric.
Except where otherwise noted, all simulations use 500 nodes, randomly distributed in a 4 km $\times$ 1 km area, and communicating via broadcast within a 0.25 km radius. 
One node is designated as a source and initial distance estimates of all nodes follow a uniform distribution between 0 and $\sqrt{17}$ km (i.e., the longest possible distance for the simulated space).

\subsection{Effect of parameters}

We begin with an empirical investigation of the design choices and parameter effects discussed in Section~\ref{sdesign}, using GABF, as defined in Section \ref{sgalg}, as an example to demonstrate the impact of these parameters  on  convergence speed. 

Progress toward convergence may be measured using the greatest overestimate $\Delta^+(t)$ and least underestimate $\Delta^-(t)$:
\begin{eqnarray}
\Delta^+(t) & = & \max\left [0,\max_{i} \left\{\Delta_i(t) \right\}\right ]\label{D+}\\
\Delta^-(t) & = & \max\left [0,-\min_{i} \left\{ \Delta_i(t) \right\} \right ].\label{D-}
\end{eqnarray}
where $\Delta_i(t)=\hat{d}_i(t)-d_i$ the distance estimation error of node $i$. Then $\Delta^+(t) = \Delta^-(t) = 0$ indicates that all distance estimates converge to their true distances at time $t$.


\begin{figure}
	\centering
	\subfigure[$\Delta^+(t)$]{
		\includegraphics[width = 1\columnwidth]{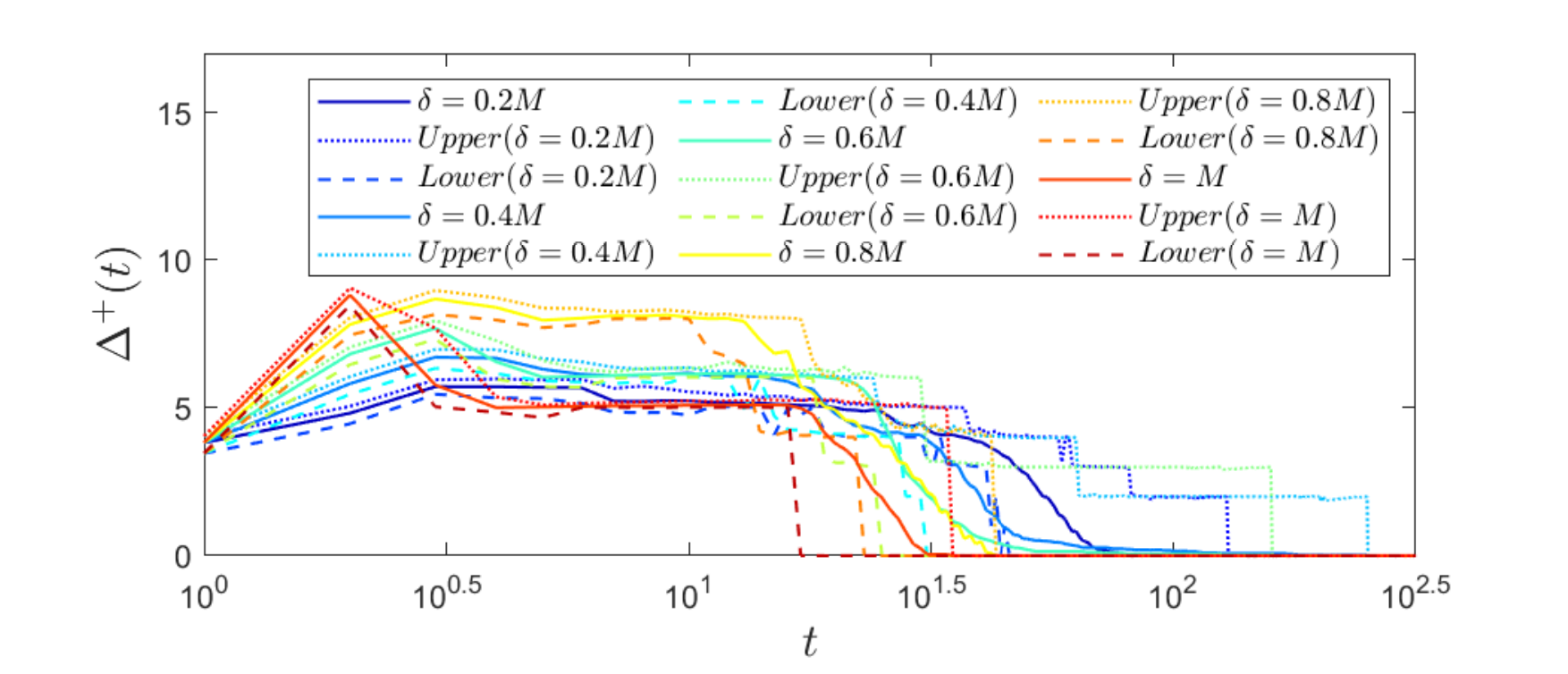}}
	\vfill
	\subfigure[$\Delta^-(t)$]{  
		\includegraphics[width = 1\columnwidth]{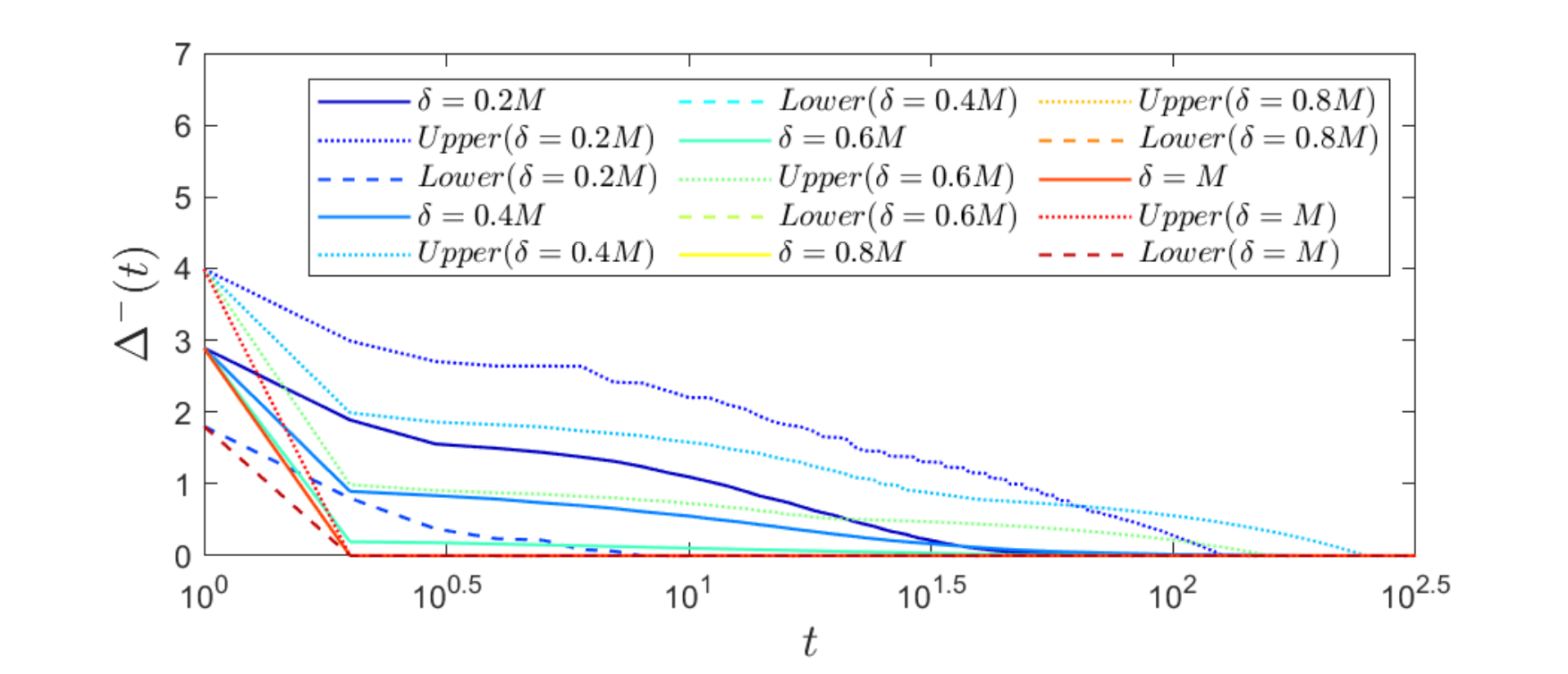}}
	\caption{Convergence time for (a) the greatest overestimate $\Delta^+(t)$ and (b) the least underestimate $\Delta^-(t)$ without perturbations, with $M = 5$, $D = 0$ and $\delta$ is varying from $0.2M$ to $M$ in steps of $0.2M$. The solid line represents the average value of 100 trials, the dotted and dashed lines represent upper and lower envelopes, respectively. 
	In (b) solid and dotted lines of $\delta = 0.8M$ and $\delta = M$ overlap, and dashed lines of $\delta = 0.4M$, $\delta = 0.6M$, $\delta = 0.8M$ and $\delta = M$ overlap.}
	\label{fig:delta}
\end{figure} 

We start with $\delta$, which controls how quickly $M$ is reached.
Figure~\ref{fig:delta} shows the results of 100 runs using $M = 5$, $D = 0$, and $\delta$ varying from $0.2M$ to $M$ in steps of $0.2M$.
The average $\cD(\cG)$ is 18.8.
With a fixed $M$, both $\Delta^+(t)$ and $\Delta^-(t)$ converge slower with a smaller $\delta$, since a smaller $\delta$ means estimates exceed $M$ later.
The corresponding average convergence times are 61.2, 54.1, 38.1, 32.6 and 24.1.
 

\begin{figure}
	\centering
	\subfigure[$\Delta^+(t)$]{
		\includegraphics[width = 1\columnwidth]{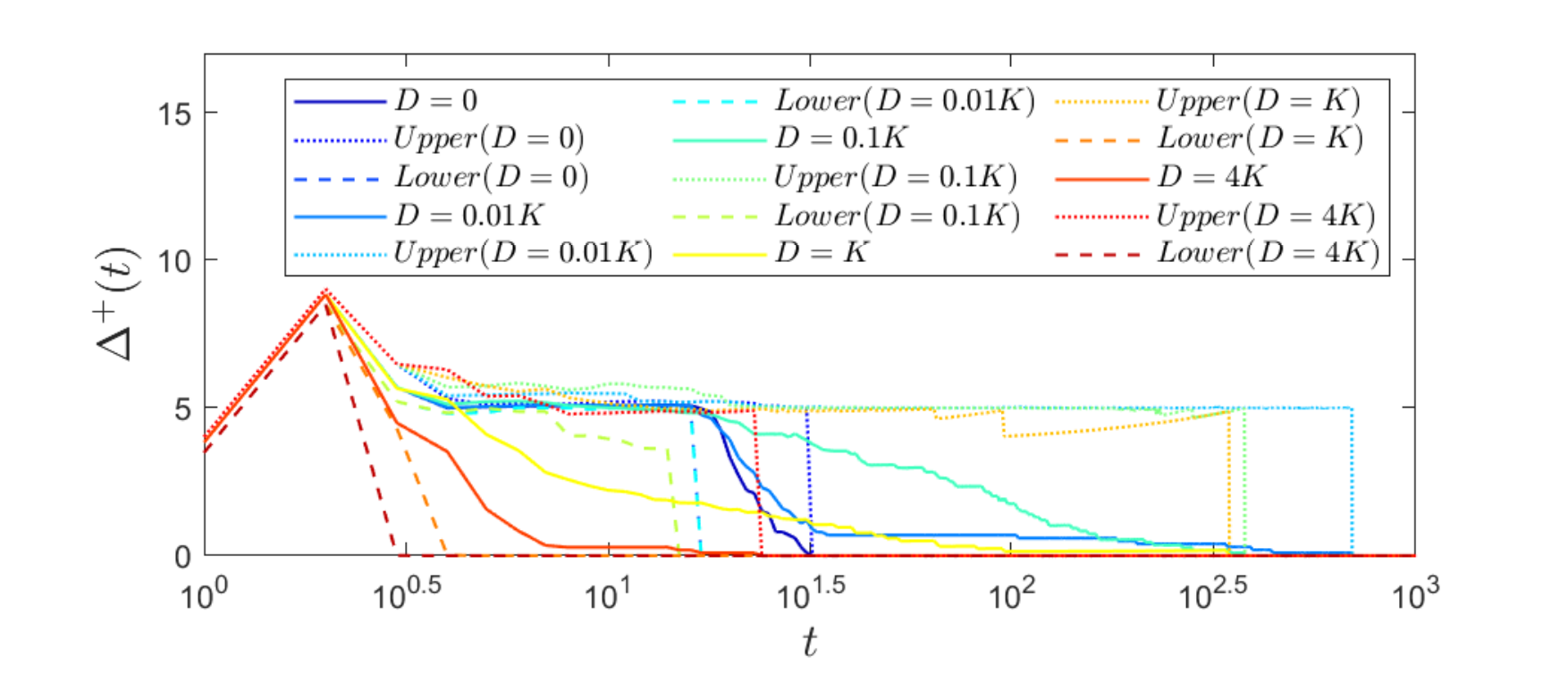}}
	\vfill
	\subfigure[$\Delta^-(t)$]{  
		\includegraphics[width = 1\columnwidth]{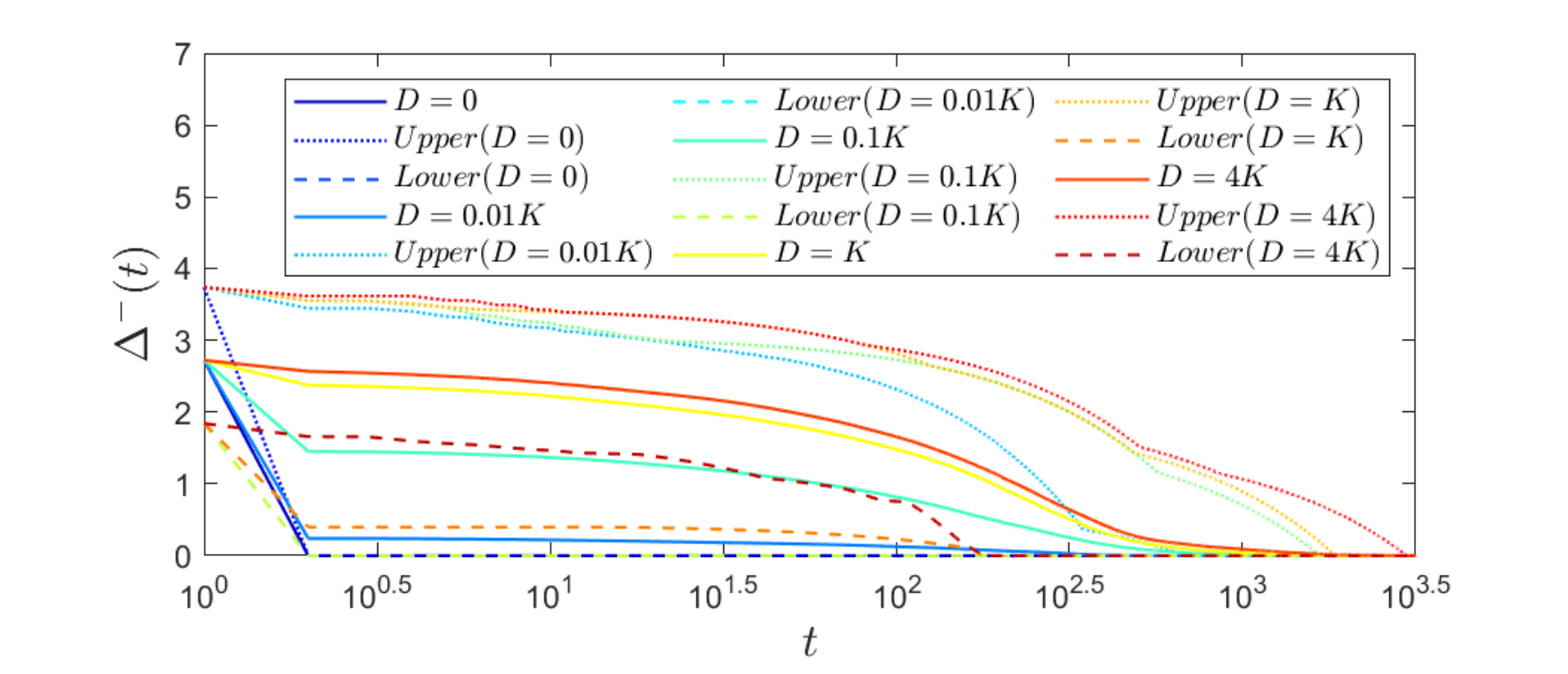}}
	\caption{Convergence time for (a) the greatest overestimate $\Delta^+(t)$ and (b) the least underestimate $\Delta^-(t)$ without perturbations, with $\delta = M = 5$ and $D = 0, 0.01K, 0.1K, K$ and $4K$, using $K = (\cD(\cG) + \cD(\cG^-) - 2)\epsilon$. The solid line represents the average value of 100 trials, the dotted and dashed lines represent upper and lower envelopes, respectively. 
	In (b) dashed lines of $D = 0, 0.01K$ and $0.1K$ overlap.}
	\label{fig:D}
\end{figure} 

The dead zone value $D$, on the other hand, has opposing effects on $\Delta^+(t)$ and $\Delta^-(t)$.
Figure~\ref{fig:D} shows the results of 100 runs using $\delta = M = 5$ and $D = 0, 0.01K, 0.1K, K$ and $4K$, using $K = (\cD(\cG) + \cD(\cG^-) - 2)\epsilon$.
In these simulations, the average value of $\epsilon$ is $2.5 \times 10^{-3}$ km and the average values of $\cD(\cG)$ and $\cD(\cG^-)$ are 17.9 and 26.8, respectively.
For $\Delta^+(t)$, a large $D$ (e.g., $D = 4K$) or a small $D$ (e.g., $D = 0$) accelerates convergence. 
In the former case GABF acts more like ABF in which case $\cD(\cG)$ will converge within $\cD(\cG) - 1$ rounds \cite{TAC}. In the latter case the second bullet of (\ref{raising}) will be frequently invoked such that underestimates will be eliminated more quickly with a large $M$, and thus $T$ defined in (\ref{tstar}) becomes smaller. 
Both of these phenomena are seen in these simulations, with the average convergence time of $\Delta^+(t)$ being 23.2, 67.9, 99.3, 34.0 and 6.5. 
As for $\Delta^-(t)$, A large $D$ always has a negative impact on the convergence speed of $\Delta^-(t)$ since the behavior of GABF is more like ABF in this situation, where the \emph{rising value problem} \cite{TAC,mo2018robust} will cause the underestimates rise very slowly.
Here, the average convergence time is 2, 49.3, 310.0, 535.9 and 631.7.
Note that these underestimates are more vulnerable to the change of $D$: while the time to convergence of $\Delta^+(t)$ is roughly 60 rounds faster by increasing $D$ from $0.1K$ to $K$, that of $\Delta^-(t)$ may be hundreds of rounds slower under such a change. 
Overall convergence time is thus regulated by $\Delta^+(t)$ for small $D$ and by $\Delta^-(t)$ for large $D$, with the joint average convergence time of 23.2, 67.9, 310.0, 535.9 and 631.7.

\begin{figure}
	\centering
	\subfigure[$\Delta^+(t)$]{
		\includegraphics[width = 1\columnwidth]{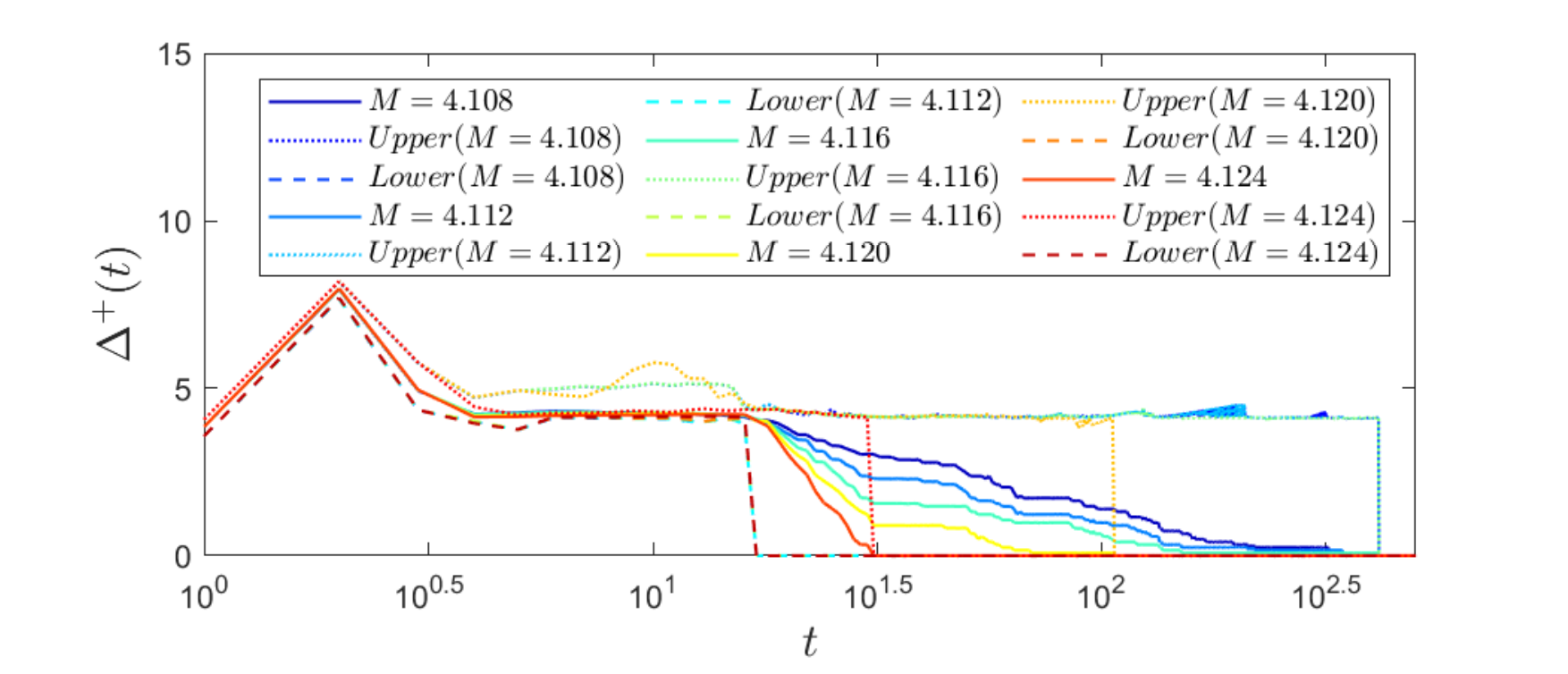}}
	\vfill
	\subfigure[$\Delta^-(t)$]{  
		\includegraphics[width = 1\columnwidth]{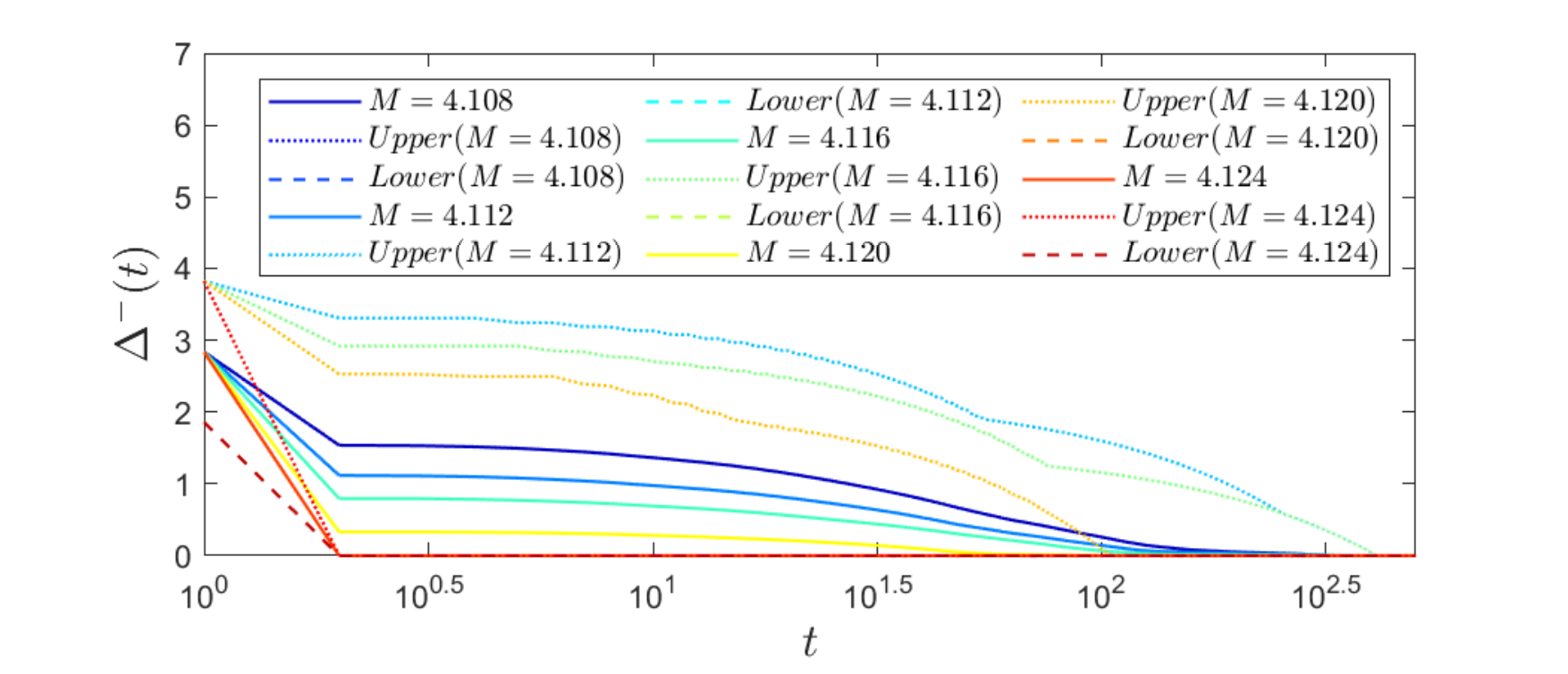}}
	\caption{Convergence time for (a) the greatest overestimate $\Delta^+(t)$ and (b) the least underestimate $\Delta^-(t)$ without perturbations, with $D = 0$, $\delta = M$, and $M$ varying from 4.108 to 4.124. The solid line represents the average value of 100 trials, the dotted and dashed lines represent upper and lower envelopes, respectively. 
	In (b) dashed lines of all different M overlap, and dotted lines of $M = 4.108$ and $4.112$ overlap.}
	\label{fig:M}
\end{figure} 

Finally, the impact of $M$ may be separated out from the other two parameters by setting $D=0$ and $\delta=M$.
In this condition, small changes in $M$ can result in large changes in convergence rate.
Figure~\ref{fig:M} illustrates this for 100 runs of GABF, with $D=0$, $\delta=M$, and $M$ increasing from 4.108 to 4.124 in steps of 0.004.
In these simulations, the average $\cD(\cG)$ is 18.5.
Here, convergence unconditionally improves with higher $M$: the average convergence times of $\Delta^+(t)$ are 92.9, 72.6, 55.1, 32.0 and 23.7 rounds, while those of $\Delta^-(t)$ are 87.3, 63.0, 43.5, 14.4 and 2 rounds.
Even though $T_i$ defined in (\ref{ti}) satisfies $T_i = 3$ for all different $M$, as we have set $\delta = M$, overestimates disappear more quickly with larger $M$ because a larger $\delta$ helps the time $T$ defined in (\ref{tstar}), after which all states are overestimates becoming smaller by (\ref{eq:ude}) in Lemma \ref{luin}. 
Underestimates converge more quickly once $M$ is greater than both the largest true distance and initial distance estimate. In this case, all nodes acquire overestimates in the first round by invoking the second bullet of (\ref{raising}),
and thus the underestimates converge in only 2 rounds.

\subsection{Robustness against persistent perturbations}

\begin{figure}
	\centering
	\subfigure[$\Delta^+(t)$ and $(\cD(\cG) - 1)\epsilon$]{
		\includegraphics[width = 0.94\columnwidth]{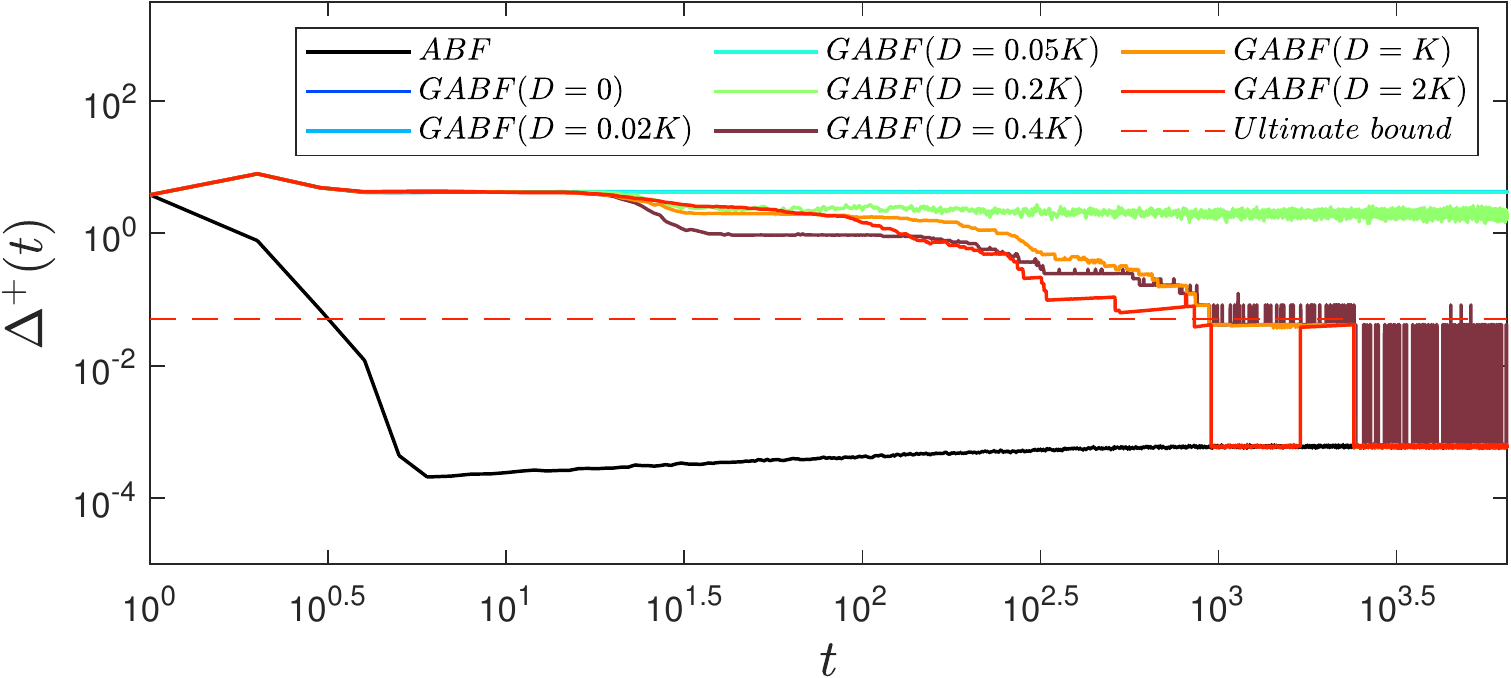}}
	\vfill
	\subfigure[$\Delta^-(t)$ and $(\cD(\cG^-) - 1)\epsilon$]{  
		\includegraphics[width = 0.94\columnwidth]{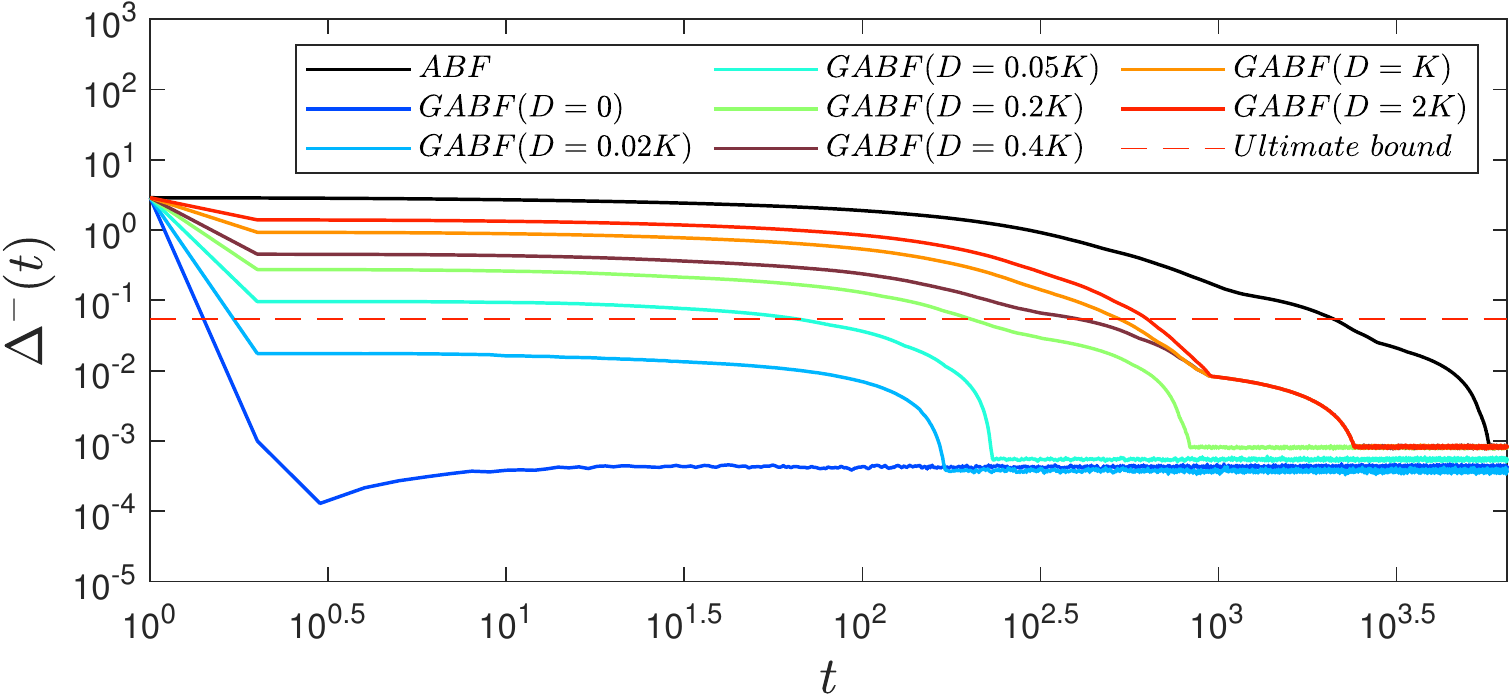}}
	\caption{Robustness against perturbation for GABF with various values of $D$, compared over 100 trials with ultimate bounds and with ABF: (a) mean values of $\Delta^+(t)$ and mean value of ultimate bound $(\cD(\cG) - 1)\epsilon$, and (b) mean values of $\Delta^-(t)$ and mean value of ultimate bound $(\cD(\cG^{-}) - 1)\epsilon$. Edge lengths are perturbed by measurement errors uniformly distributed between 0 and $e_{\min}$. Parameters for GABF are set as $\delta = M = \sqrt{17}$, $D = 0, 0.02K, 0.05K, 0.2K, 0.4K, K$ and $2K$, where $K = (\cD(\cG) + \cD(\cG^-) - 2)\epsilon$. In (a), lines of $D = 0, 0.02K$ and $0.05K$ overlap.}
	\label{fig:perturbation}
\end{figure} 

As discussed in Section~\ref{sdesign}, robustness against perturbation should be controlled primarily by parameter $D$.
In particular, $(\cD(\cG)-1)\epsilon$ and $(\cD(\cG^-)-1)\epsilon$ are the ultimate bounds of $\Delta^+(t)$ and $\Delta^-(t)$ under perturbations, respectively \cite{TAC} (i.e., $L_1W(L_2,\cD(\cG)-1)$ and $L_1W(L_2,\cD(\cG^-)-1)$ in Theorem \ref{tpre}, respectively), and thus $(\cD(\cG) + \cD(\cG^-) - 2)\epsilon$ is the minimum value of $D$ to guarantee the robustness of GABF under perturbations. 

Figure \ref{fig:perturbation} illustrates this for 100 runs of GABF, comparing this ultimate boundedness with ABF and with GABF using $\delta = M = \sqrt{17}$ and $D = 0, 0.02K, 0.05K, 0.2K, 0.4K, K$ or $2K$, where $K = (\cD(\cG) + \cD(\cG^-) - 2)\epsilon$, $\epsilon = 0.05e_{\min}$.
Perturbation is injected as asymmetric noise in the estimated $e_{ij}$, such that measurement errors $\epsilon_{ij}(t)$ defined in (\ref{vare}) follow a uniform distribution between 0 and $0.05e_{\min}$ in each round. 
In these simulations, the average value of $e_{\min}$ defined in (\ref{ep}) is $2.9 \times 10^{-3}$ km, $\cD(\cG) = 18.6$ and $\cD(\cG^-) = 20.0$ on average. 

The results show the tradeoffs in $\Delta^+(t)$ versus $\Delta^-(t)$ with GABF.
For $\Delta^-(t)$, ABF is constrained by the \emph{rising value problem} \cite{TAC,mo2018robust} such that $\Delta^-(t)$ needs a much longer time than $\Delta^+(t)$ to drop below than its ultimate bound.
With GABF, lower values of $D$ increase the speed of convergence, with $D=0$ achieving the fastest time.
For $\Delta^+(t)$, on the other hand, ABF converges extremely quickly, while GABF does not converge at all for low values of $D$.
In this case, GABF with $D = 0, 0.02K, 0.05K, 0.2K$ and $0.4K$ will not be ultimately bounded, while the average time for ABF and GBAF with $D = K$ and $2K$ to drop below the ultimate bounds follows 859.3, 196.5 and 300.1 rounds. Further, the average time for ABF and GABF with $D = K$ or $2K$ to reach the bottom is 5750 and 2406 rounds, respectively. 
Combining both, we find that GABF outperforms ABF under perturbations when the dead zone value $D$ is equal to or slightly larger than the minimum required value defined in (\ref{D}).
Thus, when the general spreading block is under perturbation, $D$ should be set as $(W(L_2,\cD(\cG^-) - 1) + W(L_2,\cD(\cG) - 1))L_1\epsilon$ defined in (\ref{D}) of Lemma \ref{tpre} in order to guarantee robustness and meanwhile attain a fast convergence speed. Observe though the floor is much below the theoretical ultimate bound, and even with $ D=.04K $, $ \Delta^+(t) $ though persistently rising from the floor rises only up to the unltimate bound.

\subsection{Non-Euclidean Distance}

\begin{figure}
	\centering
	\subfigure[The spreading block of (\ref{spreading})]{
		\includegraphics[width = 0.8\columnwidth]{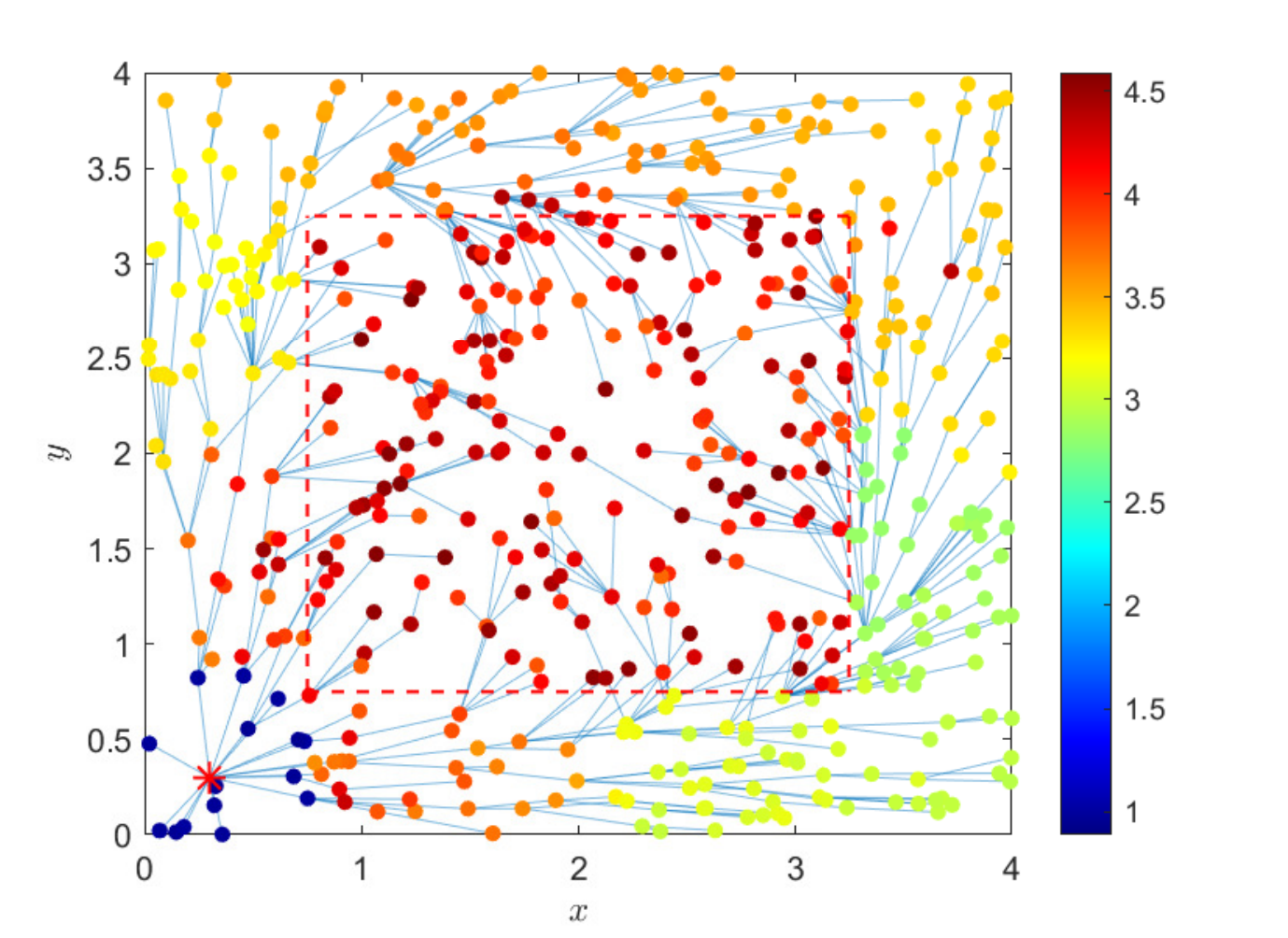}}
	\vfill
	\subfigure[The more general spreading block of (\ref{gengenG})-(\ref{raising})]{  
		\includegraphics[width = 0.8\columnwidth]{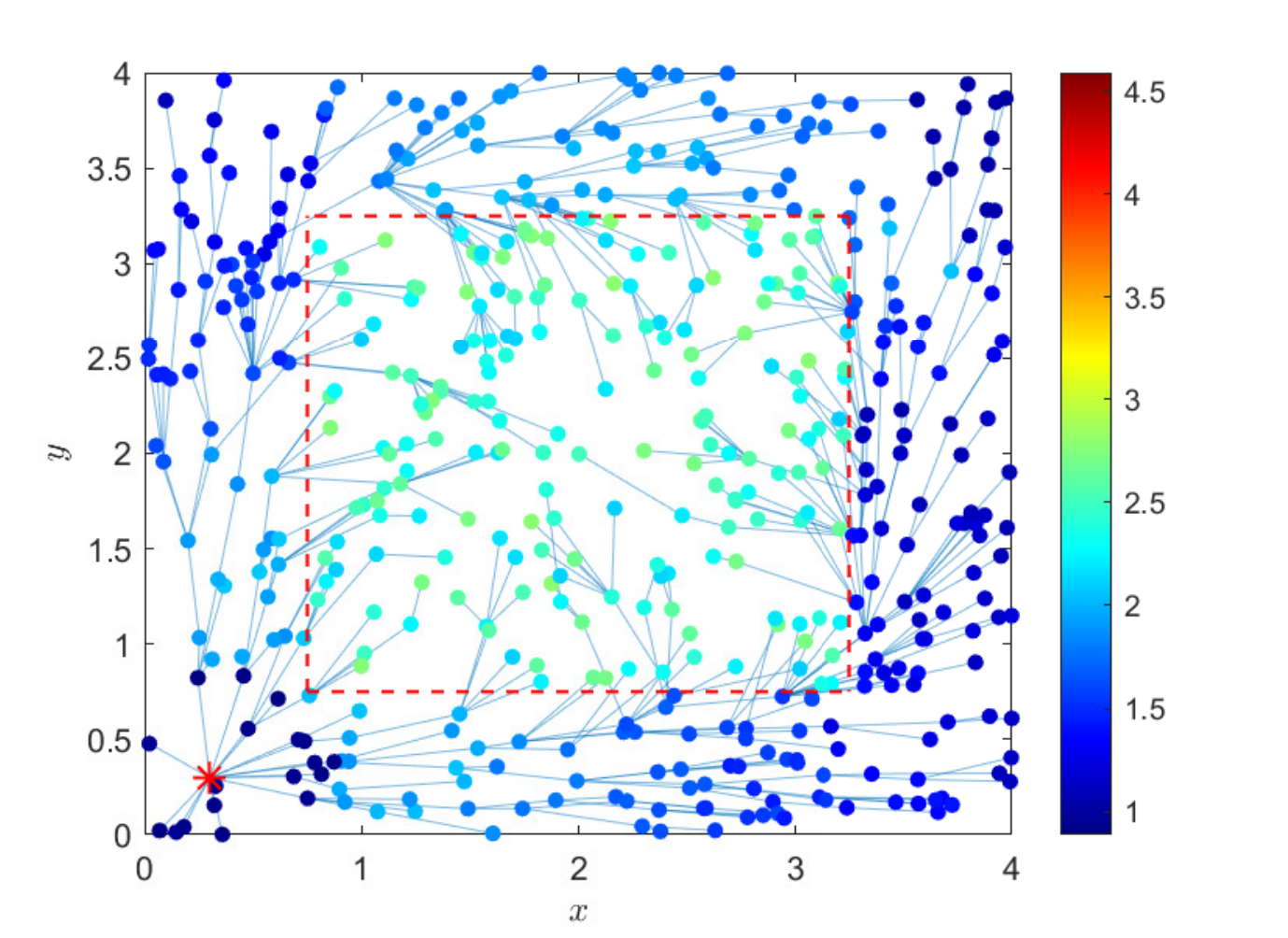}}
	\caption{In this example, 400 nodes are randomly located in a $4 \times 4$ $\mathrm{km}^2$ field. There is a source at the red asterisk located at (0.3, 0.3), and the middle of the field is a $2.5 \times 2.5$ $\mathrm{km}^2$ radiation zone. Color represents degree of contamination, with a logarithmic scale. While both spreading block and the general spreading block can achieve the shortest available path, the fast convergence of the general spreading block greatly reduces total contamination.}
	\label{fig:nonEuclidean}
\end{figure} 

Finally, we illustrate how the general spreading block can accommodate non-Euclidean distance metrics. 
Figure~\ref{fig:nonEuclidean} shows example of a  nonlinear $f(\cdot,\cdot)$ in (\ref{spreading}), for nodes to compute paths minimizing exposure to a hazard. In this scenario, 500 nodes are randomly distributed in a $4 \times 4$ $\mathrm{km}^2$ field, and communicate over a 0.6 km radius.  A source is located at (0.3, 0.3). In the middle of the area, there exists a $2.5 \times 2.5$ $\mathrm{km}^2$ \emph{radiation zone} centered at (1.95, 1.95). Define $\mathcal{M}$ as the set of nodes in the radiation zone, a node $i$ is radioactive if $i \in \mathcal{M}$ or $i$ has ever taken a radioactive node as its constraining node.  $f(\xh_k(t),e_{ik})$ in (\ref{spreading}) obeys $f(\xh_k(t),e_{ik}) = \xh_k(t) + e_{ik}, k \notin \mathcal{M}$, where $e_{ik}$ is the edge length between node $i$ and $k$. When $k \in \mathcal{M}$, $f(\xh_k(t),e_{ik}) = h(\xh_k(t) + 1000e_{ik})$, where $h(a) = a^{1.5}$ if $a > 1$ and $a$ otherwise. Further, $s_i$ defined in (\ref{spreading}) is 0 if $i$ is a source and $\infty$ otherwise. In each round a node $i$ will receive 100$\sim$120 units of radiation dose if it is radioactive and 0$\sim$1 unit otherwise. 
Figure~\ref{fig:nonEuclidean}(a) shows the result of using the spreading block defined in (\ref{spreading}), while Figure~\ref{fig:nonEuclidean}(b) shows the result from setting $D = 0, \delta = M > x_{\max}$ in the general spreading block. 
In both cases, nodes outside the radiation zone  never cross the zone due to the high cost, and nodes inside the zone  take the shortest path to exit the zone. However, the degree of contamination is greatly reduced when using the general spreading block with appropriately chosen parameters, due to the much faster time of convergence to a safe path.

\section{Conclusion}\label{sconc}
We have improved a general algorithm for spreading information across a network of devices by making it resilient to  perturbations and by removing a prior boundedness assumption.  
This algorithm, a key building block for aggregate computing and applicable to a wide range of distributed systems, has parameters that remove the rising value problem that appears in some of its special cases, such as ABF. 
Unlike ABF, however, the general spreading algorithm covers a much wider class of uses and application, such as dealing with non-Euclidean distance metrics. 
We have proven global uniform asymptotic stability for this algorithm and provide ultimate bounds in face of persistent network disturbances using  an additional Lipschitz condition. Notably, the ultimate bounds depend only on the largest perturbation and structural network properties. Finally, we provide design guidelines for the three new parameters, demonstrating how algorithm parameters have competing effects on performance.

These results are a crucial stepping stone in our long term goal of determining stability conditions for feedback interconnections of aggregate computing blocks, using possibly new small gain theorems, \cite{smallgain}, or equivalent techniques, \cite{pass}, like the passivity theorem and its variants, \cite{IQC}.
Progress in this program has broad applicability for the engineering of resilient distributed systems.






\bibliographystyle{IEEEtran}
\bibliography{preprint}

\appendix
\noindent
{\bf Proof of Lemma \ref{sluin}:}
From Definition \ref{dur}, (\ref{implies}) still holds if (\ref{genG}) is replaced by (\ref{genGt}). Thus $\cU(t)$ can be nonempty only on a single contiguous time interval commencing at $t_0$. 

We prove (\ref{suin}) by induction. As it holds for $t = t_0$, suppose (\ref{suin}) holds for some $t \geq t_0$. Consider $ i\in \cU(t+1) $ with $\xh_i(t + 1) = \xh_{\min}(t + 1)$. Then $\cU(t)\neq \emptyset$ and $j$ the  constraining node of $i $ is in $\cU(t)$. If $ i\in \cE(t+1) $  in Definition \ref{dAE} then from Definition \ref{dcons}, $ j=i $. From the induction hypothesis
and (\ref{strg})
\begin{eqnarray}
\xh_{i}(t + 1) & = & \xh_{\min}(t + 1) \nonumber \\
& \geq & \xh_{j}(t) + \delta \geq \xh_{\min}(t)+\min\{\sigma,\delta\} \nonumber \\
&\geq& \xh_{\min}(t_0) + \min\{\sigma,\delta\}(t + 1 - t_0). \nonumber
\end{eqnarray}
If $i \in \cA(t + 1)\cap \cU(t+1)$, then, $ i\notin \cS(t+1) $,  i.e. $ \xh_i(t+1)\neq s_i. $ From the induction hypothesis and (\ref{genGt}) we have
\begin{align*}
	\xh_{i}(t + 1)  = & \xh_{\min}(t + 1) = f(\xh_{j}(t),e_{ij}(t)) \\
	& \geq  \xh_{j}(t) + \sigma \geq \min\{\sigma,\delta\}(t + 1 - t_0). \nonumber
\end{align*}
Further with (\ref{Xmax}) and (\ref{tm}), (\ref{exX}) follows.

\noindent
{\bf Proof of Lemma \ref{tover}:}
If  $\cR(t) \neq \emptyset$, then $ \exists ~t_5,t_6 $ such that $\forall t_0 \leq t_5 \leq t\leq t_6$ and $\cR(t_5) = \cS(t_5)$. As $\xh_i(t_5) = s_i$ for all $i \in \cS(t_5) = \cR(t_5)$, from (\ref{lem}) the result holds for $t = t_5$.

 Suppose $\xh_i(t) \geq X_i$ for some $t_5 \leq t < t_6$ and  all $i \in \cR(t)$. Consider any $i \in \cR(t+1)$. From Definition \ref{dur}, either $i \in \cS(t+1)$ in which case the rsult holds, or or $i$ is constrained by some $j \in \cR(t).$ If $i \in \cA(t+1)$, then by the induction hypothesis, $\xh_j(t) \geq X_j$. As $i \notin \cS(t+1)$, there follows:
\begin{eqnarray}
\xh_{i}(t + 1) &=& f(\xh_{j}(t), e_{ij}(t)) \nonumber \\
&\geq& f(X_j,e_{ij}^-) \label{el} \\
&\geq& X_i \label{ki} 
\end{eqnarray}
where (\ref{el}) uses $e_{ij}(t) \geq e_{ij}^-$ for all $t$ and the fact that $f(\cdot,\cdot)$ is increasing in each argument, and (\ref{ki}) uses (\ref{stat}).
If $i \in \cE(t + 1)$, then $i$ is its own constraining node and $i \in \cR(t)$. Thus by our induction hypothesis, $\xh_i(t) \geq X_i$. From (\ref{strg}), 
\begin{equation}
\xh_{i}(t + 1) \geq \xh_{i}(t) + \delta 
>  X_i.\label{bkq}
\end{equation}

\noindent
{\bf Proof of Lemma \ref{tlbound}:}
 Consider nodes $n_0,n_1,\cdots,n_T$ such that $n_0 \in \cS_{\infty}^-$, and for all $i \in \{0,\dots,T - 1\}$, $n_i$ is a true constraining node of $n_{i+1}$ in $\cG^-$. Each node in $\cG^-$ is in one such sequence. As from Definition \ref{shrunken}, $T \leq \cD(\cG^-) - 1$, the result will follow if
\begin{equation}\label{induc}
x_{n_i} - X_{n_i} \leq W(L_2,i)L_1\epsilon, ~\forall i \in \{0,\cdots,T\}.
\end{equation}
Evidently, $x_{n_0} \leq s_{n_0} = X_{n_0}$. Suppose (\ref{induc}) holds for some $i \in \{0,\cdots, T-1\}$. As $n_i$ and $n_{i+1}$ are neighbors in both $\cG$ and $\cG^-$, $n_i$ is a true constraining node of $n_{i+1}$ in $\cG^-$ and $x_{n_i} \leq X_{n_i} + W(L_2,i)L_1\epsilon$ by our induction hypothesis, 
\begin{flalign}
x_{n_{i+1}} &\leq f(x_{n_i},e_{n_in_{i+1}}) \nonumber\\
&\leq f(X_{n_i} + W(L_2,i)L_1\epsilon, e_{n_in_{i+1}}) \nonumber \\
&\leq f(X_{n_i},e_{n_in_{i+1}}) + L_2W(L_2,i)L_1\epsilon\label{ul2} \\
&= f(X_{n_i},e_{n_in_{i+1}}^- + \epsilon) + L_2W(L_2,i)L_1\epsilon \label{eq}\\
&\leq f(X_{n_i},e_{n_in_{i+1}}^-) + L_1\epsilon +  L_2W(L_2,i)L_1\epsilon \label{em} \\
&= X_{n_{i+1}} + W(L_2,i+1)L_1\epsilon \label{ist}
\end{flalign}
where (\ref{ul2}) uses (\ref{aa}), (\ref{eq}) uses (\ref{shru}),
(\ref{em}) uses (\ref{ul1}), and (\ref{ist}) uses the fact that $n_{i}$ is a true constraining node of $n_{i+1}$ in $\cG^-$. 

\noindent
{\bf Proof of Lemma \ref{tpre}:}
Consider any $i \in \cF_{L+1}$. 
Because of 
  (\ref{raising}) and (\ref{strg}),  there is a $t' < t \leq t' + T_{L+1}^-$ such that $i \in \cA(t)$. This is so as $ i\in \cE(t) $ implies $ \xh_i(t+1)\geq \xh_i(t)+\delta $ and at some time  in the interval $ (t',t' + T_{L+1}^-] $, $ \xh_i(\cdot)>M$. From Lemma \ref{lF}, there is a $j \in \cF_L$ that is a true constraining node of $i$ in $\cG$. Thus $\xh_{j}(t - 1) \leq x_j + W(L_2,L)L_1\epsilon$ by (\ref{pre}). Then  
\begin{flalign}
\xh_{i}( t ) 
&=\min\left\lbrace \min_{k\in \cN(i)}\left\lbrace f\left (\xh_k(t-1), e_{ik}(t-1) \right ) \right\rbrace, s_i  \right\rbrace \nonumber \\
&\leq f(\xh_{j}(t-1),e_{ij}(t-1)) \nonumber \\
&\leq f(x_j + W(L_2,L)L_1\epsilon,e_{ij}+\epsilon) \label{usin} \\
&\leq f(x_j,e_{ij}) + L_2W(L_2,L)L_1\epsilon + L_1\epsilon \label{l1l2} \\
&= x_i + W(L_2,L+1)L_1\epsilon \label{ne1} 
\end{flalign}
where (\ref{usin}) uses (\ref{vare}), (\ref{ep}) and (\ref{pre}), (\ref{l1l2}) uses (\ref{ul1}) and (\ref{aa}). Similarly, as (\ref{pre}) holds for all $ t\geq T^-+t_0 $ for all $ j\in \cF_L $,
\begin{equation}\label{ne2}
\xt(t+1)\leq x_i + W(L_2,L+1)L_1\epsilon.
\end{equation}
As $t  > t_0 + T^-$,  (\ref{tall}) implies that $\xh_k(t) \geq X_k$ for all $k \in V$. As $f(\cdot,\cdot)$ is monotonically increasing in  both its arguments and $X_i \leq s_i$, (\ref{recm}) implies that 
\begin{flalign}
\xt_{i}(t+1)=\min\left\lbrace \min_{k\in \cN(i)}\left\lbrace f\left (\xh_k( t), e_{ik}( t) \right ) \right\rbrace, s_i  \right\rbrace \nonumber \\  \geq \min\left\lbrace\min_{k\in \cN(i)}\left\lbrace f\left (X_k, e_{ik}^- \right )  \right\rbrace, s_i \right \rbrace \nonumber  = X_i,
\end{flalign}
 i.e. $ \left [X_i , x_i + W(L_2,L+1)L_1\epsilon\right ]$ contains both $ \xh_i(t) $ and $ \xt_i(t+1) $. Then (\ref{D}) and Lemma \ref{tlbound} yield
\begin{eqnarray}
|\tilde{x}_i(t + 1) - \xh_{i}(t)| &\leq& |x_i + W(L_2,L+1)L_1\epsilon - X_i| \nonumber \\
& \leq & W(L_2,\cD(\cG^-) - 1)L_1\epsilon + \nonumber \\
&& W(L_2,L + 1)L_1\epsilon 
\leq D, \nonumber
\end{eqnarray}
i.e, $\xh_{i}(t + 1) = \tilde{x}_{i}(t + 1)$. An  induction proves the result. 

\noindent
{\bf Proof of Theorem \ref{bound}:}
From  Lemma \ref{tlbound} and (\ref{tall})
\begin{equation}\label{upper}
\xh_i(t) - x_i \geq -W(L_2,\cD(\cG^-) - 1)L_1\epsilon ~\forall ~t\geq t_0+T^-, 
\end{equation}
proving the lower bound on $ \xh_i(t) - x_i $ implicit in the theorem statement. To prove the upper bound
we will first show that 
\begin{equation}\label{T0}
\xh_i(t)\leq x_i=x_i+W(L_2,0), ~\forall i\in \cF_0, ~t\geq t_0+T^-+T_0^-.
\end{equation}
Then the repeated application of Lemma \ref{tpre} will prove  that
\[\xh_i(t) - x_i \leq W(L_2,\cD(\cG) - 1)L_1\epsilon ~\forall ~t \geq t_0 + T^- + \sum_{i = 0}^{\cD(\cG)-1}T^-_i \]
and thus the theorem.

Consider $ i\in \cF_0. $
As $i \in \cE(t)$ implies $\xh_i(t+1)\geq \xh_i(t)+\delta$ from (\ref{raising}), (\ref{strg}) and (\ref{tmi}), there is a $t_0+T^- < t \leq t_0+T^- + T_{0}^-$ such that $i \in \cA(t)$. As $\cF_0\subset \cS_\infty$,  from (\ref{sta}) 
\begin{flalign}
\xh_{i}(t)  = \xt_{i}(t) 
& \leq s_i = x_i. \label{mz}
\end{flalign}
As $t  > t_0 + T^-$, it follows from (\ref{tall}) that $\xh_k(t) \geq X_k$ for all $k \in V$. As $f(\cdot,\cdot)$ is monotonically increasing in  both its arguments and $X_i \leq s_i$, we obtain
\begin{flalign}
\xt_{i}(t+1) = \min\left\lbrace\min_{k\in \cN(i)}\left\lbrace f\left (\xh_k(t), e_{ik}(t) \right ) \right\rbrace, s_i \right\rbrace  \nonumber \\
\geq \min\left\lbrace\min_{k\in \cN(i)}\left\lbrace f\left (X_k, e_{ik}^- \right ) \right\rbrace, s_i \right\rbrace = X_i \label{fo}
\end{flalign}
where (\ref{fo}) uses (\ref{recm}). Therefore, $ \left [X_i, x_i\right]$ contains both $\xh_i(t)$ and $ \xt_i(t+1)$. Then (\ref{D}) and Lemma \ref{tlbound} yield
\begin{flalign}
|\xt_i(t + 1) - \xh_{i}(t)| &\leq |x_i - X_i| \nonumber \\
&= W(L_2,\cD(\cG^-)-1)L_1\epsilon \nonumber \\
&< D \nonumber
\end{flalign}
From (\ref{gengenG}-\ref{raising}), $\xh_{i}(t + 1) = \tilde{x}_{i}(t + 1)$. Am  induction proves (\ref{T0}). 
%
\begin{IEEEbiography}
	[{\includegraphics[width=1in,height=1.25in,clip,keepaspectratio]{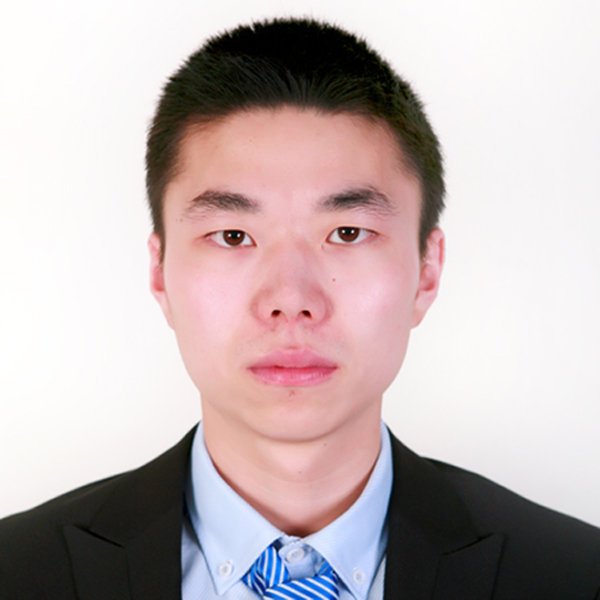}}]
	{\bf Yuanqiu Mo} was born in Yangzhou, China, in 1991. He received the Ph.D. degree in electrical and computer engineering at the University of Iowa, in 2019. He is currently a postdoc with the Westlake University.
	His research interests include distributed algorithm design and stability theory. Dr. Mo has been awarded the 2018 CDC Outstanding Student Paper Award. He was a finalist of the Young Author Award in the IFAC World Congress 2020.
\end{IEEEbiography}

\begin{IEEEbiography}
	[{\includegraphics[width=1in,height=1.25in,clip,keepaspectratio]{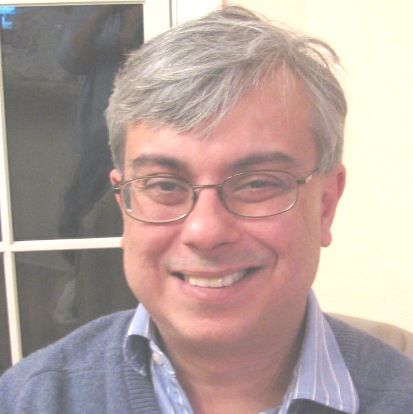}}]
	{\bf Soura Dasgupta, (M'87, SM'93, F'98)} was born in 1959 in Calcutta, India. 
	He received the B.E. degree (Hons. I) in Electrical Engineering from the University of Queensland (Australia) in 1980, and the Ph.D. in Systems Engineering from the Australian National University, in 1985.  He is currently F. Wendell Miller Distinguished Professor in the Department of Electrical and Computer Engineering at the University of Iowa, U.S.A  and holds a visiting appointment in the Shandong Academy of Sciences.
	
	In 1981, he was a Junior Research Fellow  at the Indian Statistical Institute, Calcutta. He has held visiting appointments at the University of Notre Dame, University of Iowa, Universite Catholique de Louvain-La-Neuve, Belgium, Tata Consulting Services, Hyderabad,  the Australian National University and National ICT Australia.
	
	From 1988 to 1991, 1998 to 2009 and 2004 to 2007 he respectively served as an Associate Editor of the IEEE TRANSACTIONS ON AUTOMATIC CONTROL, IEEE Control Systems Society Conference Editorial Board, and the IEEE TRANSACTIONS ON CIRCUITS AND SYSTEMS- II. He is a co-recipient of the Gullimen Cauer Award for the best paper published in the IEEE TRANSACTIONS ON CIRCUITS AND SYSTEMS in the calendar years of 1990 and 1991, a past Presidential Faculty Fellow, a past subject editor for the International Journal of Adaptive Control and Signal Processing, and a member of the editorial board of the EURASIP Journal of Wireless Communications. In 2012 he was awarded the University Iowa Collegiate Teaching award. In the same year he was selected by the graduating class for  excellence in teaching and commitment to student success. From 2016-18 he was a 1000 Talents Scholar in the People’s Republic of China. 
	
	His research interests are in Controls, Signal Processing, Communications and Parkinson's Disease. He was elected a Fellow of the IEEE in 1998. 	
\end{IEEEbiography}

\begin{IEEEbiography}
	[{\includegraphics[width=1in,height=1.25in,clip,keepaspectratio]{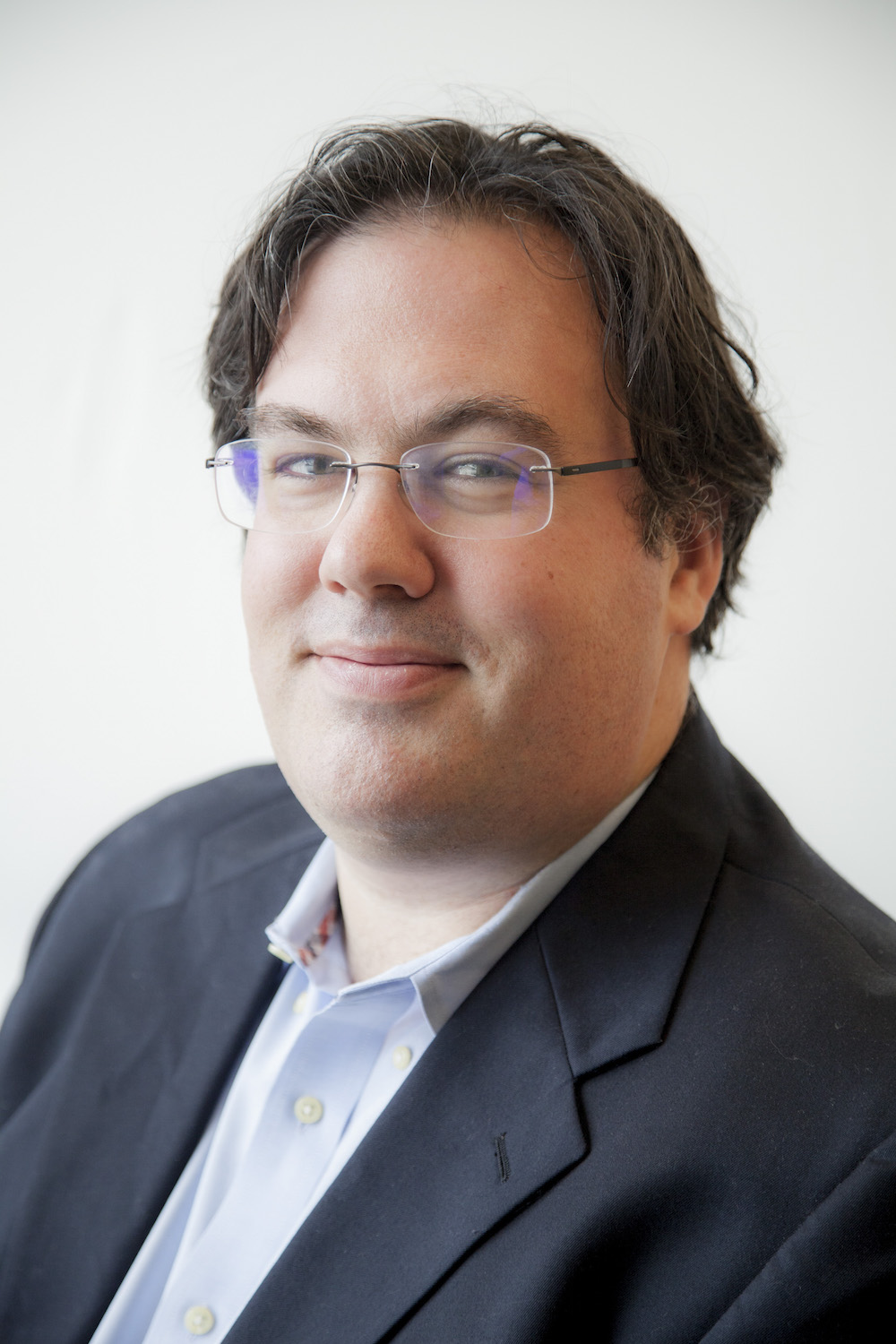}}]
	{\bf Jacob Beal} is a scientist at Raytheon BBN Technologies in Cambridge, Massachusetts.
	His research focuses on the engineering of robust adaptive systems,
	particularly on the problems of aggregate-level modeling and control for
	spatially distributed systems like pervasive wireless networks, robotic swarms,
	and natural or engineered biological cells. Dr. Beal received a PhD in electrical
	engineering and computer science from MIT. He is an associate editor of ACM
	Transactions on Autonomous and Adaptive Systems, is on the steering committee
	of the IEEE International Conference on Self-Adapting and Self-Organizing
	Systems (SASO), and is a founder of the Spatial Computing Workshop series.
	He is a Senior Member of IEEE. Contact him at jakebeal@ieee.org
\end{IEEEbiography}

\end{document}